\documentclass[reprint,pra,a4paper]{revtex4-1}

\usepackage{graphicx}
\usepackage{amssymb,amsmath,amsfonts}
\usepackage{color}
\usepackage[normalem]{ulem}
\usepackage{braket}
\usepackage{amsthm}
\newtheorem{prop}{Proposition}
\newtheorem{thm}{Theorem}
\newtheorem*{thm_faithful_repeated}{Theorem~\ref{thm_faithful}}
\newtheorem*{thm_isolated_repeated}{Theorem~\ref{thm_isolated}}
\newtheorem*{thm_diamond_repeated}{Theorem~\ref{thm_diamond}}

\usepackage{algorithm}
\usepackage[noend]{algpseudocode}
\usepackage{textcomp}
\usepackage{mathrsfs}

\usepackage{graphicx,amsmath,amsfonts,algorithm}
\floatname{algorithm}{Protocol}
\usepackage{graphics}
\usepackage{appendix}
\usepackage{amssymb}

\usepackage{amsthm}
\usepackage{complexity}

\usepackage{color}
\usepackage{bbm,bm}

\usepackage[normalem]{ulem}
\usepackage{makeidx}  % allows for indexgeneration
\usepackage{url}

\usepackage{float}
\usepackage{algorithm}
\usepackage{multirow}
\usepackage{natbib}
\usepackage{amsthm}

\usepackage{graphicx}
\usepackage{array}

\usepackage{verbatim}
\usepackage{textcomp}
\usepackage{mathtools}

\newcommand{\be}{\begin{equation}}
\newcommand{\ee}{\end{equation}}
\newcommand{\ben}{\begin{eqnarray}}
\newcommand{\een}{\end{eqnarray}}
\newcommand{\bF}{\begin{figure}}
\newcommand{\eF}{\end{figure}}
\newcommand{\bi}{\begin{itemize}}
\newcommand{\ei}{\end{itemize}}

\def\ket#1{| #1 \rangle}

\def\bra#1{\langle #1 |}

\usepackage[colorlinks=true,citecolor=blue,linkcolor=blue,urlcolor=blue]{hyperref}

\newcommand{\RomanNumeralCaps}[1]
    {\MakeUppercase{\romannumeral #1}}

\begin{document}
\title{Subtleties of witnessing quantum coherence in non-isolated systems}

\author{George~C.~Knee}
\email{gk@physics.org}
\affiliation{Department of Physics, University of Warwick, Coventry CV4 7AL, United Kingdom}
\author{Max Marcus}
\affiliation{Department of Physics, University of Warwick, Coventry CV4 7AL, United Kingdom}
\author{Luke~D.~Smith}
\altaffiliation[Current address: ]{School of Chemistry, University of Leeds, Leeds LS2 9JT, United Kingdom}
\author{Animesh~Datta}
\email{animesh.datta@warwick.ac.uk}
\affiliation{Department of Physics, University of Warwick, Coventry CV4 7AL, United Kingdom}
\date{\today}                                           % Activate to display a given date or no date

\begin{abstract}
Identifying non-classicality unambiguously and inexpensively is a long-standing open challenge in physics. The No-Signalling-In-Time protocol was developed as an experimental test for macroscopic realism, and serves as a witness of quantum coherence in isolated quantum systems by comparing the quantum state to its completely dephased counterpart. We show that it provides a lower bound on a certain resource-theoretic coherence monotone. We go on to generalise the protocol to the case where the system of interest is coupled to an environment. Depending on the manner of the generalisation, the resulting witness either reports on system coherence alone, or on a disjunction of system coherence with either (i) the existence of non-classical system-environment correlations or (ii) non-negligible dynamics in the environment. These are distinct failure modes of the Born approximation in non-isolated systems.
\end{abstract}

\maketitle
\section{Introduction}
Quantum mechanics continues to revolutionise our understanding of light and matter on ever larger scales and in ever more complex systems. Its counter-intuitive predictions have long been the subject of skepticism, which has in turn spurred on the development of fundamental tests such as Bell's inequality~\cite{Bell1964,ClauserHorneShimony1969}. This test involves making measurements on each of a pair of spatially separated  quantum systems. If the measurements are rapid enough and the separation is large enough, the possible correlations between measurement results are bounded according to any `local hidden-variable' theory. The predicted violation of this bound by quantum mechanics, and the experimental demonstration thereof, promises far more than just a refutation of the classical point of view: the emergent field of quantum information science and technology is broadly predicated on exploiting these `nonclassical' correlations. The experimental methodology is invariably to isolate systems to such an extent that their quantum character is readily apparent.

Some of these technologies directly leverage Bell's approach, making them `device independent' - meaning that one does not even need to believe in quantum mechanics in order to trust in the security (for example) of a secret key distribution protocol~\cite{BarrettHardyKent2005}. Furthermore, if quantum mechanics is assumed, violation of Bell's inequality witnesses (i.e. is sufficient, but not necessary to infer) the existence of entangled quantum states~\cite{Terhal2000}. There exists a hierarchy of states ranging from  Bell-inequality-violating, through entangled and discordant states, with different classes being exploited by different quantum technologies~\cite{AdessoBromleyCianciaruso2016}. 

While the earliest tests of Bell's inequality date back decades, \cite{FreedmanClauser1972,AspectDalibardRoger1982}, systematically closing loopholes in such experiments involved great technological effort and has only been comprehensively achieved very recently, placing the failure of a classical explanation beyond all reasonable doubt~\cite{GiustinaVersteeghWengerowsky2015,HensenBernienDreau2015,ShalmMeyer-ScottChristensen2015}. The strictest test requires exercising precise and rapid control over widely separated, highly isolated physical systems -- infeasible in most physical scenarios. This is true even of engineered systems such as quantum computers~\cite{LaddJelezkoLaflamme2010} as it is of natural systems such as bio-molecular complexes~\cite{WildeMcCrackenMizel2010}. For instance, in 2007, Engel \emph{et al.} presented their evidence for quantum coherence in the excited state dynamics of the Fenna-Matthews-Olson (FMO) light-harvesting complex by way of a long-lived oscillatory signature revealed by two-dimensional electron spectroscopy~\cite{EngelCalhounRead2007}. Although it would be desirable to adapt such experiments (which rely on incidental signatures of quantum coherence and are therefore subject to alternative, classical explanations~\cite{PerlikLincolnSanda2014}) to leverage Bell's test for a more robust confirmation of non-classicality, this seems only a very distant possibility. The encapsulated Bacteriochlorophyll pigments which compose the FMO, for instance, are separated by only a few Angstroms: light traverses such distances in less than an attosecond, making a strict Bell test infeasible with current technology~\footnote{By contrast a recent test over a $\sim$1km separation~\cite{HensenBernienDreau2015} enjoyed a light transit time of a few $\mu$s}. Furthermore it is not necessarily desirable to isolate such systems, since the interaction with the environment is often the subject of great interest, for example playing a potentially crucial role in energy transport~\cite{PlenioHuelga2008,MohseniRebentrostLloyd2008,CarusoChinDatta2009,OReillyOlaya-Cast2014}. Hence the need for protocols tailored to non-isolated systems on very small scales, such as those developed in this paper.

A modification of Bell's test due to Leggett and Garg (LG)~\cite{LeggettGarg1985} concerns correlations across time rather than across space. Instead of local causality they coined the term `macrorealism': the composite view that a sufficiently large system occupies exactly one of its possible states at any given moment, and that this state may be determined in a non-invasive manner. These assumptions codify classical physics but are contradicted by most interpretations of quantum theory. There exist at least three alternative readings of macrorealism~\cite{MaroneyTimpson2014,HermensMaroney2018,AllenMaroneyGogioso2017}: Here we adopt the `eigenstate-mixture' interpretation that is most amenable to experimental test and arguably most relevant~\cite{KneeKakuyanagiYeh2016} to models of dynamical wavefunction collapse~\cite{BassiGhirardi2003}. The original proposal from LG called for determination of several two-time correlation functions -- a daunting challenge in the laboratory. 

Recently, a refined protocol termed `No-Signalling-In-Time'  (NSIT) was developed which constitutes a simpler and more effective test of macrorealism~\cite{LiLambertChen2012,KoflerBrukner2013,ClementeKofler2016,KneeKakuyanagiYeh2016}. As we show below, it may also be thought of as a state coherence witness for isolated systems. The NSIT condition, as the LG inequality before it, is predicated on the negligible effect of earlier measurements on later ones, and forms the basis of this article. The condition is essentially an expression of the classical Kolmogorov consistency conditions relating a probability density over a set of temporally separated measurement outcomes to its marginal distributions -- the quantum violation of which was actually predicted by LG in their original paper~\cite{LeggettGarg1985}. The more recent moniker of NSIT reflects the (dis)analogy with the Bell inequality: famously, quantum correlations are able to surpass those of any local theory but are insufficient to allow signalling across space. In the temporal scenario, quantum states (and, one might add, many plausible hidden-variable models) do not conform to such a compromise, and are quite capable of signalling in time. In fact the spatial correlations achievable in quantum theory are bounded by the so-called Tsirelson's bound~\cite{Cirelson1980}, and are weaker than the most general non-signalling correlations~\cite{PopescuRohrlich1994}. An equivalent bound has been argued to apply in the temporal case to the set of divisible quantum channels~\cite{LePollockPaterek2017}.

The NSIT condition has been shown to be necessary and sufficient for macrorealism~\cite{ClementeKofler2015}, and Fine's theorem ---which states that Bell inequalities form a necessary and sufficient condition for the existence of a single joint probability distribution over all measurement outcomes~\cite{Fine1982} --- does not carry over straightforwardly to the temporal case~\cite{ClementeKofler2016,Halliwell2017}. Coupled with a more favourable experimental outlook, this has lead some to eschew the LG inequality in favour of the NSIT condition~\cite{KneeKakuyanagiYeh2016,WangKneeZhan2017}.
%MR does not imply a diagonal rho. If [M',\rho]=0 then MR picture is possible despite non-diagonal rho.

Experimental violations of macrorealism have been found in a variety of well-isolated optical and solid-state quantum systems; for a review, see~\cite{EmaryLambertNori2014} and subsequent experiments~\cite{RobensAltEmary2016,KneeKakuyanagiYeh2016,FormaggioKaiserMurskyj2016,WangKneeZhan2017,KatiyarBrodutchLu2017}. 
However, inferences about the existence of quantum coherence drawn from such tests make implicit assumptions about the coupling of the system to its environment. For systems such as the FMO complex, however, the LG inequality~\cite{WildeMizel2011} (with a simple dynamical model) and the NSIT condition~\cite{LiLambertChen2012} (with a more sophisticated dynamical model) have only been calculated theoretically. The quantum or classical question has also been theoretically investigated with similar tools in other nano-structured, open quantum systems~\cite{LambertEmaryChen2010}.

We provide the theoretical framework necessary to unambiguously infer quantum coherence in non-isolated systems experimentally. We go beyond isolated systems and specify features other than quantum coherence that can trigger violations. We begin in Section~\ref{sec_classicalisation} by introducing fast and slow variants of a classicalisation operation, key to the rest of the paper. In Section~\ref{sec_isolated} we describe how this operation can be implemented in the laboratory to witness quantum coherence of isolated systems, making a connection to the resource theory of coherence. In Section~\ref{sec_general} we define three experimental protocols relating to non-isolated systems, quantify their cost and deduce which states and processes are able to trigger nonzero values, refining and supplementing pre-existing results. A short discussion on device independence follows in Section~\ref{sec_device}, whereafter we end with concluding remarks in Section~\ref{sec_conclusion}.

\section{Classicalisation operation}
\label{sec_classicalisation}
A key component of the NSIT protocol is the ability to transform a system, described by an unknown density operator $\rho=\sum_{ij} \rho_{ij}|i\rangle \langle j | $, acting on a $d$-dimensional Hilbert space $\mathscr{H}_S$, into a particular classically-equivalent diagonal state. In this context the classical state in question is the unique diagonal density operator exhibiting the same probability distribution as $\rho$, when both are expressed in a pre-ordained and privileged basis $\{|i\rangle\}_i$.  The classicalisation operation  is written $\Gamma(\rho)$. Mathematically, it can be thought of as the result of masking (i.e. multiplying through the element-wise Hadamard product) $\rho$ with $\textrm{diag}(1,1,\ldots,1)$: preserving the diagonal entries but destroying the off-diagonal ones. Note that coherence is here unambiguously taken to mean the nonzero value of at least one of these off-diagonal entries, rather than any other notion such as the coherence of classical waves~\cite{ScholesFlemingChen2017}. 

For the remainder of this paper we drop the `quantum' from `quantum coherence' for brevity.

It is important to realise that coherence is a basis-dependent notion. The basis  $\{| i \rangle\}_i$ with respect to which coherence is witnessed or measured is defined by $\Gamma$, and often arises naturally depending on the physical scenario. This is in analogy with the Bell inequality and entanglement, which relies on a given bi-partitioning of the Hilbert space (usually set by appeal to special relativity and spatial separation).  
%One might wish to consider a basis of spatially distinct `sites' or of energy eigenstates, for example. 
For NSIT, a naturally preferred basis may be indicated by a dominant decoherence channel determined by the form of the coupling to the environment~\cite{MeznaricClarkDatta2013}, or the ability to measure only in specific bases such as energy. Importantly, $\Gamma$ is a valid quantum operation and therefore should be implementable in the laboratory. It may be achieved in at least the following two ways: through artificial dephasing or through a blind measurement. 

Artificial dephasing involves arranging for a random distribution of phase factors $e^{i\theta_j}$ to be applied 
\begin{align}
\Gamma(\rho)=\int d\vec{\theta} p(\vec{\theta})e^{i\sum_j\theta_j |j\rangle\langle j|} \rho e^{-i\sum_k\theta_k |k\rangle\langle k|}
\label{gamma1}
\end{align}
such that the mean of  $e^{i[\theta_j-\theta_k]}$ is at the origin of the complex plane $\forall j,k$. Note that the integrand is the conjugation of $\rho$ with a unitary matrix which is diagonal in the preferred basis. The choice $p(\vec{\theta})=\prod_jp(\theta_j)$ with $p(\theta_j)=\delta(\theta_j-0)/2+\delta(\theta_j-\pi)/2$
%$\prod_ip_i(\theta_i)$ with $p_i(\theta_i=0)=p_i(\theta_i=\pi)=50\%$
achieves the desired effect -- for a qubit this represents a $50\%$ chance of having a phase flip or not. The net operation has been achieved experimentally e.g. by randomising the phase of path-encoded photonic qudits~\cite{WangKneeZhan2017}. 

A blind measurement~\cite{SchildEmary2015}, on the other hand, is simply a measurement in the preferred basis for which the result is discarded: the post-measurement state is not conditioned on the measurement outcome, but is instead subject to the average effect of the transformations corresponding to the different outcomes  
\begin{align}
\Gamma(\rho)=\sum_i |i\rangle  \langle i | \rho|i\rangle \langle i |.
\label{gamma2}
\end{align}
This amounts to taking the system density matrix apart and putting it back together again without the quantum coherences $\rho_{ij}, i\neq j$ so that only populations $\rho_{ii}$ remain~\cite{LiLambertChen2012}. Here the operators for the different measurement outcomes $|i\rangle\langle i|$ are mutually orthogonal: by contrast other studies have sought to test quantumness through the use of weak measurements~\cite{GogginAlmeidaBarbieri2011,MancinoSbrosciaRoccia2018}, where the operators strongly overlap and reveal less information about the system. We also assume that there are $d$ measurement operators, meaning the measurement corresponds to a non-degenerate observable and is of the von-Neumann type~\cite{SchildEmary2015}.

Witnessing quantumess in non-isolated systems depends crucially on the timescale on which classicalisation is achieved.  In the first instance, $\Gamma$ is implemented dynamically -- in `real time' -- much faster than other characteristic timescales in the experiment (particularly the timescales of the environment to which the system might be coupled). In other cases, $\Gamma$ may be implemented on a much slower timescale, often in a piecewise fashion. This is especially the case when the blind measurement implementation of $\Gamma$ is sought, but when measurements fail to leave the system in a state compatible with the measurement outcome. Examples include fluorescence readout of the qubits encoded in spin states of NV centres~\cite{JelezkoGaebelPopa2004} or in energy level of trapped ions~\cite{HartyAllcockBallance2014}, free induction decay in nuclear magnetic resonance~\cite{Jones2011}, absorptive detection of single photons~\cite{Hadfield2009} as well as stimulated processes in nonlinear spectroscopy~\cite{Yuen-ZhouAspuru-Guzik2011}. Although a measurement of the preferred basis is performed -- the populations of the various classical states are inferred -- it is not through a process which can be modelled by projective measurement operators. Take for example the Kraus operators $K_i=|\phi_i\rangle\langle i|$: the various outcomes have the correct probabilities $\textrm{tr}(K_i\rho K_i^\dagger)=\textrm{tr}(|i\rangle\langle i|\rho)$ but the post measurement state is $|\phi_i\rangle\neq|i\rangle$. %Previously, ideal non-demolition measurements have been achieved in a spin ensemble with the use of ancillary qubits~\cite{MoussaRyanCory2010,KneeSimmonsGauger2012}. By contrast, 
The solution that we concentrate on in this paper is to infer the full set of probabilities and then re-prepare the appropriately weighted mixture of classical states from a fiducial state, resulting in $\Gamma(\rho)$.

For isolated systems, fast and slow classicalisation operations are trivially equivalent. The equivalence breaks down when we consider non-isolated systems; it is inadequate to restrict the quantum operation $\Gamma$ to $\mathscr{H}_S$: it must be expanded to the full Hilbert space of system and environment.

\section{NSIT protocol for an isolated system and connection to the resource theory of coherence}
\label{sec_isolated}
\begin{figure}
\centering
\includegraphics[width=\columnwidth]{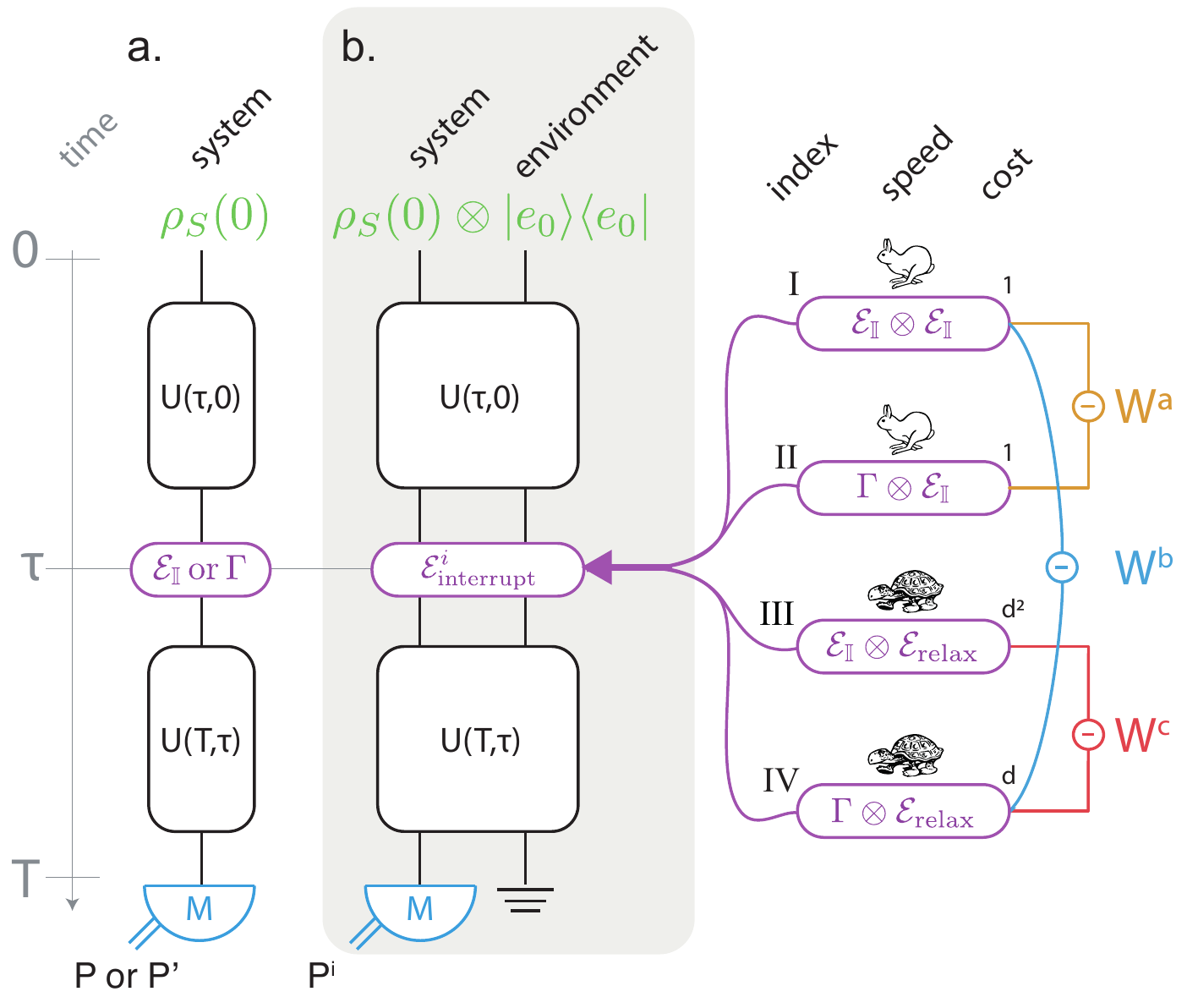}
\caption{Procedures for extracting coherence witnesses (time runs from top to bottom): a. In the usual NSIT protocol the system is considered well isolated b. Our generalisations include an explicit environment, with each interruption experiment (fast operations $i=\RomanNumeralCaps{1},\RomanNumeralCaps{2}$ and slow operations $i=\RomanNumeralCaps{3},\RomanNumeralCaps{4}$) carrying a cost in terms of the number of sub-experiments required. $\mathcal{E}_\mathbb{I}$ is the identity channel. \label{fig1}}
\end{figure}

Consider an isolated system initialised in a fiducial state described by a density operator $\rho_S(t=0)$. It is then allowed to evolve under either its natural Hamiltonian or through active control fields, resulting in a test state $\rho_S(t=\tau)$. The NSIT protocol is based on testing the effect of $\Gamma$, acting at time $t=\tau$, on the test state. For an isolated system, the witness is defined
\begin{align}
W^{\textrm{isolated}}(\rho_S(\tau),M'):&=P-P'\nonumber\\
&=\textrm{tr}_S(M'\left[\rho_S(\tau)-\Gamma(\rho_S(\tau))\right]),
\label{witness}
\end{align}
where $P'$ ($P$) is the probability of the measurement outcome corresponding to the positive operator $M$  occurring at a later time $t=T$, when the earlier $\Gamma$ operation was (was not) performed. $M'$ is the effective measurement operator at $\tau$, which is related to the actual measurement operator $M$ by $M'=U^\dagger(T,\tau)MU(T,\tau)$ -- see Fig~\ref{fig1}a. 
Recently, temporal correlations have been cast into a hierarchy based on the violation of LG's inequality, temporal steerability and temporal non-separability, given the satisfaction of NSIT, i.e. $W^{\textrm{isolated}}=0$~\cite{KuChenLambert2017}. By contrast we seek signatures of $W\neq0$. Note that we leave our witnesses as signed quantities: the absolute value is taken in some of the prior literature~\cite{LiLambertChen2012,SchildEmary2015}.

In order to witness coherence in this way it is necessary for the classicalisation operation $\Gamma$ to be available to the experimenter and  (in the first instance) for the operation to be trusted. In Section~\ref{sec_device} we discuss the possibility of removing the trust, resulting in a device-independent protocol.

It is necessary for the difference in states at $\tau$ to translate into a divergence in measurement outcome probabilities at $T$ for the witness to be triggered -- in other words, $M'$ must be well chosen. If $M$ is a measurement in the preferred basis, then $M'$ is chosen only by $U(T,\tau)$ and the conclusion of Smirne \emph{et al.}~\cite{SmirneEgloffDiaz2018} applies: the multi-time statistics cannot be considered classical if the dynamics generates coherences and subsequently turns them into populations. In fact, considering the best choice for $M'$ makes the connection between tests of macrorealism and the resource theory of coherence.

If $M'$ is chosen optimally, the following inequality is saturated
\begin{align}
|W^{\textrm{isolated}}(\rho_S(\tau),M')|&\leq \max_{M'} \textrm{tr}(M'[\rho_S(\tau)-\Gamma(\rho_S(\tau))]) \nonumber\\
&=: R(\rho_S(\tau))/2.
\label{measure}
\end{align}
The functional
\begin{align}
R(\rho)=||\rho-\Gamma(\rho)||_{\textrm{tr}}
\label{vulnerability_of_coherence}
\end{align}
is twice the trace distance between the state $\rho$ and its classical counterpart $\Gamma(\rho)$. Equivalently, it is the trace norm $||\cdot||_{\textrm{tr}}$ (or sum of singular values~\cite{Bhatia1997}) of the hollow matrix formed from $\rho$ by replacing all diagonal elements  $\langle i|\rho|i\rangle$ with zero. It is a coherence measure enjoying several attractive mathematical properties. Satisfaction of these properties -- it is zero if and only if $\rho$ is diagonal (is in the set of `free states' $\mathcal{I}$) and does not increase under a well-defined class of `free operations' -- qualifies $R(\rho)$ as a coherence monotone under a resource theory of coherence that has `dephasing covariant' operations (those that commute with $\Gamma$~\cite{MeznaricClarkDatta2013}) as the free operations~\cite{MarvianSpekkens2016}. A proof of these properties is given by Marvian and Spekkens~\cite{MarvianSpekkens2016} and also implies monotonicity under the closely related (sub)set of strictly incoherent operations~\cite{YadinMaGirolami2016}. The resource-theoretic approach to coherence is anticipated to shed new light on quantum metrology, thermodynamics, computation and cryptography, speed limits, energy transport, foundational issues and quantum technologies in general~\cite{MarvianSpekkens2016,BaumgratzCramerPlenio2014}. On the other hand, many measures exist and clear operational meanings have yet to be fully worked out. Moreover most of the monotones are not cheaply measurable and for some it is not even known how to compute them. These latter drawbacks do not apply to $R(\rho)$, increasing its relative attractiveness. 

Because $R(\rho)$ is a norm, it also satisfies the property of being convex in $\rho$. Since $\Gamma$ is a resource-destroying map~\cite{LiuHuLloyd2017}, any contractive distance between $\rho$ and $\Gamma(\rho)$ would also qualify as such a monotone. It is not known whether $R(\rho)$ is a monotone under alternative sets of free operations -- e.g. `incoherent' operations (those that do not create coherence)~\cite{BaumgratzCramerPlenio2014}. The $ l_1$-norm of coherence $C_{l_1}(\rho)=\sum_{i\neq j }|\langle i | \rho| j \rangle|$ \emph{is} a coherence monotone under incoherent operations but has no known operational meaning nor method of determination other than via full state tomography, requiring order $d^2$ experiments. 

The availability of a classicalisation operation $\Gamma$, therefore, greatly reduces the experimental costs associated with learning about certain resource monotones (compared to those monotones that require full knowledge of the state). Depending on how $\Gamma$ is implemented, as few as two experiments are required. Li \emph{et al.} introduce the idea of partial summation to further reduce the required number of experiments~\cite{LiLambertChen2012} -- a procedure we refine in Appendix~\ref{partial_sum}. $R(\rho)$ might be given the name `vulnerability of coherence' to underline its interpretation as measuring the extent to which the coherence in $\rho$ is affected by $\Gamma$. 

To see that $W^{\textrm{isolated}}(\rho_c,M)=R(\rho_c)/2=0$ for classical states $\rho_c\in\mathcal{I}$, simply note that $\Gamma(\rho_c)=\rho_c$ for any state that can be written as a diagonal density operator in the preferred basis $\rho_c=\sum_i p_i |i\rangle\langle i |$. Note that such states are convex combinations of states drawn from the preferred basis. At the other extreme, we have
\begin{thm}
\label{thm_isolated}
The maximum value of $|W^{\textup{isolated}}|$ is given by
\begin{align}
|W^{\textup{isolated}}|\leq \max_\rho R(\rho)/2\leq 1-1/d
\end{align}
where $d$ is the dimension of $\mathscr{H}_S$.
\end{thm}
\begin{proof}
This was proved in Ref.~\cite{SchildEmary2015}. For completeness we give our own explicit proof in Appendix~\ref{max_isolated}, where $\max_\rho R(\rho)/2$ is recognised as the induced trace norm distance between the identity channel and $\Gamma$.
 \end{proof} 
This maximum can be achieved for the maximally coherent state $\rho\rightarrow |+\rangle\langle +|$ with $|+\rangle=\sum_i |i\rangle/\sqrt{d}$. Due to the dependence of this upper bound on the Hilbert space dimension one can also certify a lower bound on the possibly-unknown dimension of $\mathscr{H}_S$ via $d\geq 1/(1-|W^{\textrm{isolated}}|)$. 

Having made the connection between the NSIT witness condition and the resource theory of coherence, we now proceed to generalise the witness to non-isolated systems.
% the case where the system is coupled to an environment. 

\section{NSIT for non-isolated systems}
\label{sec_general}
In most realistic experiments, the existence of an environment representing uncontrollable degrees of freedom must be acknowledged. Although one is generally free to define the system-environment divide anywhere one pleases, it is very often set by experimental limitations. For instance during the excited state dynamics in light-harvesting complexes, vibrational modes that interact with the various electronic excited states; or in spin based quantum information media, a solid-state environment containing nuclear spins can cause uncontrollable interactions and decoherence on the timescale of the system dynamics. 

In allowing for such an environment, represented by a Hilbert space $\mathscr{H}_E$ of arbitrary dimension, we will consider two pairs of interruption operations at the intermediate time $\tau$, which generalise the operations both of doing nothing to and of classicalising the system. Joint system-environment states $\rho_{SE}$ operate on the joint Hilbert space $\mathscr{H}_S\otimes\mathscr{H}_E$.
As shown in Fig~\ref{fig1}b, various probabilities $P^i$ ($i=\RomanNumeralCaps{1},\RomanNumeralCaps{2},\RomanNumeralCaps{3},\RomanNumeralCaps{4}$) are defined using a measurement at a later time $T$ of the form:
\ben
P^{i}&=&\textrm{tr}[(M\otimes \mathbb{I})(\mathcal{U}_{T,\tau}\circ\mathcal{E}^{i}_{\textrm{interrupt}}\circ\mathcal{U}_{\tau,0}(\rho_{SE}(0)))]\nonumber\\
&=& \textrm{tr}_S[M \textrm{tr}_E (\mathcal{U}_{T,\tau}\circ\mathcal{E}^{i}_{\textrm{interrupt}}\circ\mathcal{U}_{\tau,0}(\rho_{SE}(0))) ]\nonumber\\
&=& \textrm{tr}_S[M\rho^{i}_S(T)] .
\label{interruption_equations}
\een
The symbol $\circ$ stands for the concatenation of superoperators. 
Here $\rho_{SE}(0)=\rho(0)\otimes |e_0\rangle\langle e_0|$, reflecting our assumption that the initial state of system and environment is a product at $t=0$ (we will not exclude the possibility of correlations at other times, however). Here $|e_0\rangle$ is the ground state of the environment, also the equilibrium state at zero temperature. The extension to mixed states of the environment  -- such as its finite temperature thermally equilibrated state -- is straightforward and treated in Appendix~\ref{mixed_env}.
$\mathcal{U}_{t_i,t_j}(\rho_{SE}) = U(t_i,t_j)\rho_{SE} U(t_i,t_j)^\dagger$ represent unitary, joint system-environment evolutions propagating the state from $t_j$ to $t_i$. Since the measurement $M$ acts only on the system, the statistics depend only on the reduced state of the system  $\rho^i_S(T)$ with the environment traced out. 

We consider four forms that $\mathcal{E}_{\textrm{interrupt}}$ (super-operating on $\mathscr{H}_S\otimes\mathscr{H}_E$) may take. The first, minimal approach is to merely expand our operations in the usual way by including an environment which undergoes trivial evolution during the interruption. We therefore have the operations
\begin{align}
\mathcal{E}^{\RomanNumeralCaps{1}}_{\textrm{interrupt}}=& \qquad \mathcal{E}_\mathbb{I}\otimes\mathcal{E}_\mathbb{I}\\
\mathcal{E}^{\RomanNumeralCaps{2}}_{\textrm{interrupt}}=& \qquad \Gamma\otimes\mathcal{E}_\mathbb{I}
\end{align}
which we respectively title the `do nothing' and `dynamically classicalise' operations. These are interruptions where, respectively, nothing at all is performed or $\Gamma$ is applied to the system only on a timescale much \emph{faster} than the thermal relaxation of the environment $\mathcal{E}_{\textrm{relax}}(\rho)=|e_0\rangle\langle e_0|$, -- a channel acting on $\mathscr{H}_E$. Note that $\mathcal{E}_{\mathbb{I}}(\rho)=\mathbb{I}\rho\mathbb{I}=\rho$ acts on either $\mathscr{H}_S$ or $\mathscr{H}_E$: to be clear, we do not require the system and/or environment to undergo trivial dynamics overall, just that it is not actively interrupted at $\tau$. 

The second pair of interruptions are achieved piecewise on a timescale much \emph{slower} than the thermal relaxation of the environment: the net operation therefore appears as if we have intervened into the dynamics of the environment and caused it to reset (although we simply allow it to re-equilibrate to $|e_0\rangle$). Such an approach would be typical when measurements destroy the system of interest (such as photodetection) and require that the experiment be started afresh (perhaps even with a different instance of the physical system prepared in an identical state -- see the discussion in Sec.~\ref{sec_classicalisation}). The operations are 
\begin{align}
\mathcal{E}^{\RomanNumeralCaps{3}}_{\textrm{interrupt}}=& \qquad  \mathcal{E}_\mathbb{I}\otimes\mathcal{E}_{\textrm{relax}}  \label{eq:int3}\\
\mathcal{E}^{\RomanNumeralCaps{4}}_{\textrm{interrupt}}=&\qquad \Gamma\otimes\mathcal{E}_{\textrm{relax}} \label{eq:int4}
\end{align}
interruptions which we respectively title `reset environment' and `piecewise classicalise'.  We stress that these are names that reflect the effective operations that are applied, rather than doing justice to their actual implementation, which we now elaborate on. Resetting the environment only ($\RomanNumeralCaps{3}$) is the most demanding of all the interruption operations; it can be achieved by performing full state tomography at $\tau$, so that the system can be re-prepared in its reduced state long after the environment has fully relaxed to equilibrium. Piecewise classicalisation ($\RomanNumeralCaps{4}$) is simpler, and can be achieved through tomography of system populations only, followed by re-preparation of each of the appropriately weighted classical states $|i\rangle$ at $\tau$. The net operation effectively re-prepares the system in the classicalised version of the state it was in at $\tau$ while allowing the environment to re-equilibrate. The second pair of interruption experiments are expensive in the sense that they demand order $d^2$ or $d$ experiments respectively.

%The idea of interrupting the system has been employed recently to provide an operational Markov condition for quantum processes~\cite{PollockRodriguez-RosarioFrauenheim2018} -- in that case, however, the interruption is a `causal break', meaning that it does not depend on its input. None of our interruptions have this property. 

Given operations $\mathcal{E}^{i}_{\textrm{interrupt}}$, we construct three new witnesses 
\begin{align}
W^a:=P^{\RomanNumeralCaps{1}}-P^{\RomanNumeralCaps{2}}
\label{Wadef}
\end{align}
(which responds to system coherence and/or quantum system-environment correlations), 
\begin{align}
W^b:=P^{\RomanNumeralCaps{1}}-P^{\RomanNumeralCaps{4}}
\label{Wbdef}
\end{align}
(which responds to system coherence and/or coupling to a non-stationary environment) and 
\begin{align}
W^c:=P^{\RomanNumeralCaps{3}}-P^{\RomanNumeralCaps{4}}
\label{Wcdef}
\end{align}
(which responds to system coherence only). We have $|W^{a,b,c}|<1$ in all cases, although we derive tighter bounds below. Our operations and witnesses are summarised in Fig.~\ref{fig1}b. We will now proceed to show in detail why these witnesses respond to different aspects of the system and environment state and the various subtleties they encapsulate.

\subsection{$W^a$: Fast classicalisation}
\label{subsec_Wa}
\begin{figure*}
\centering
\includegraphics[width=1.8\columnwidth]{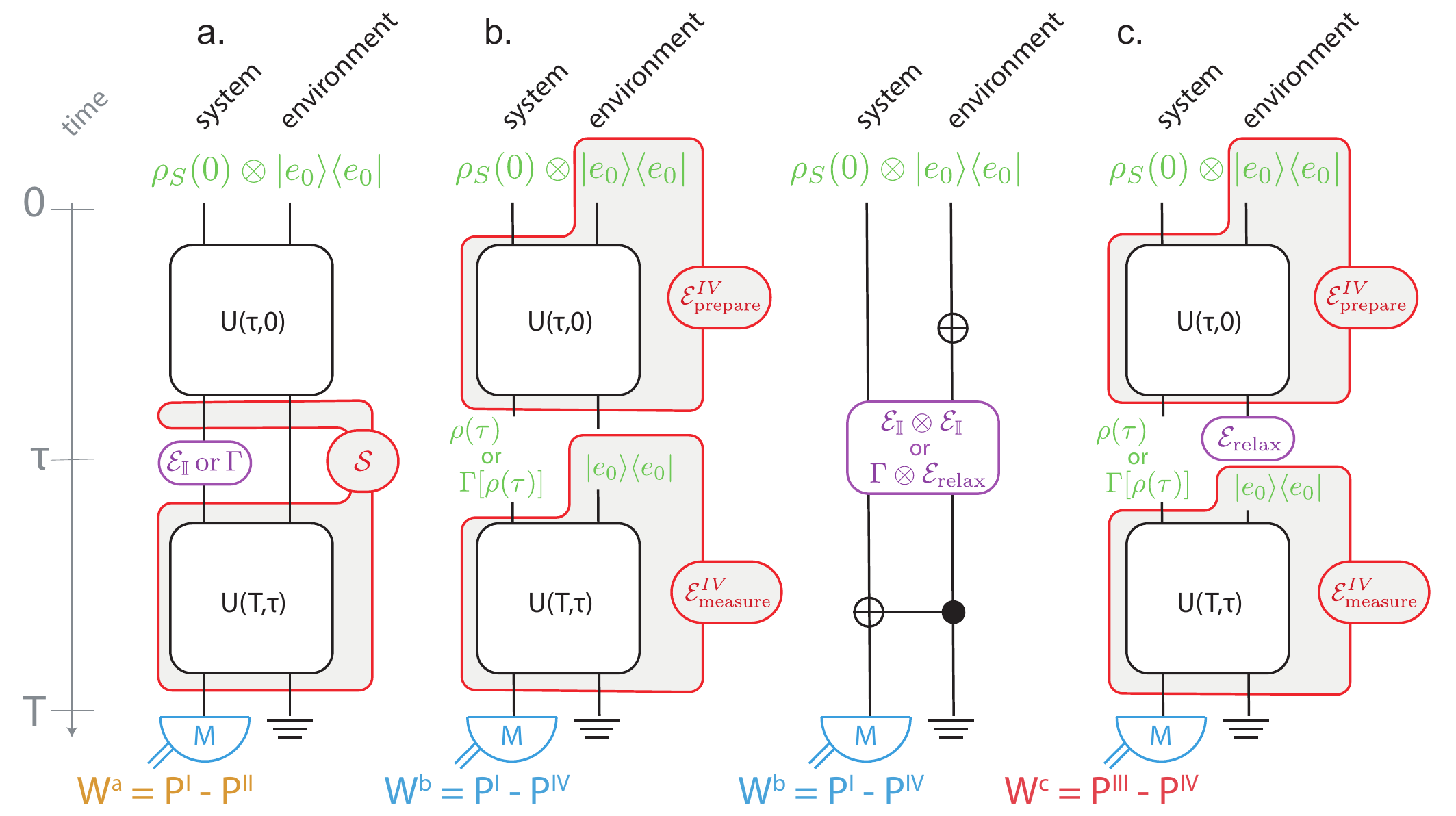}
\caption{a. In the general case of a non-isolated system, it is not possible to separate preparation and measurement of the system state at $\tau$ as is possible for an isolated system. Instead, a generalised witness compares the effect of a fixed superchannel $\mathcal{S}$ (which represents the combined effect of the correlated system-environment state at $\tau$ and their joint evolution) on an input interruption channel: either the identity channel $\mathcal{E}_{\mathbb{I}}$ or the classicalisation channel $\Gamma$ (Eq.~\eqref{gamma1} or~\eqref{gamma2}) acting on the system. The associated witness $W^a$ is sensitive to system coherence and/or quantum system-environment correlations at $\tau$. b. When the Born approximation holds, the environment stays in the same state at all times and therefore need not be actively reset (left panel). $\mathcal{E}_{\textrm{measure}}^{\RomanNumeralCaps{4}}$ is used to effectively alter the final measurement that is performed. Otherwise, when the Born approximation fails, classical models may explain the discrepancy between $P$ and $P'$ (right panel). c. The Born approximation can be enforced by performing full quantum state tomography of the system at $\tau$. %Quantum channels. $\mathcal{E}_{\textrm{prepare}}^{\RomanNumeralCaps{4}}$, which represents the combined effect of the preparation and the interaction with the environment, is used to prepare the system state at $\tau$, and . 
The associated witness $W^c$ is sensitive only to coherence in the reduced state of the system at $\tau$, but requires order $d^2$ experiments due to the implicit full state tomography required. \label{fig2} }
\end{figure*}

It is only possible to test $W^a$ if one has sufficient control to dynamically classicalise the system: that is, to apply $\Gamma$ in `real time' without stopping the experiment. This was the approach showcased in Ref.~\cite{KneeKakuyanagiYeh2016}. According to Li \emph{et al.}~\cite{LiLambertChen2012}, $W^a$ may uncover entanglement between system and environment --- a conclusion we are able to refine somewhat.

Coupling to an environment generally leads to a departure from unitary dynamics of the system state, and is commonly treated with the completely positive (CP) map formalism (otherwise known as the operator sum representation, or quantum operations formalism~\cite{NielsenChuang2004}). Since the joint system-environment state does \emph{not} factorise at $\tau$ for interruptions $\RomanNumeralCaps{1}$ and $\RomanNumeralCaps{2}$, however, we must adopt a more general approach than is allowed by the CP map formalism; for example the `superchannel' formalism developed by Modi~\cite{Modi2012}. Moreover, we are not able to claim to witness properties of the system alone (not even of the reduced state at $\tau$). We shall then see that $W^a$ is sensitive to any changes in the joint system-environment state (and their subsequent evolution) induced by applying $\Gamma$ to the system.
 
Modi introduces the superchannel $\mathcal{S}$ as an object with six indices which transforms a quantum channel (superoperating on $\mathscr{H}_S$) into a quantum state (operating also on $\mathscr{H}_S$). It represents everything in the protocol apart from the part of the interruption channel that acts on the system; namely the system-environment joint state at $\tau$ as well as their unitary dynamics -- see Fig.~\ref{fig2}a.. That is, it acts as 
\ben
\mathcal{S}\bullet\mathcal{E}&=&\rho(T)\nonumber\\
\sum_{r',r'',s',s''}\mathcal{S}_{rr'r''ss's''}\mathcal{E}_{r'r''s's''}&=&\rho_{rs}(T).
\een
Upon defining 
\ben
\mathcal{S}_{rr'r''ss's''}=\sum_{\alpha,\epsilon,\beta}U_{r\epsilon r'\alpha}(T,\tau)\rho^{SE}_{r''\alpha s''\beta}(\tau)\bar{U}_{s\epsilon s'\beta}(T,\tau),\nonumber\\
\een
(with complex conjugates denoted with a bar), we are able to write, using Eqs.~\eqref{interruption_equations} and~\eqref{Wadef}
\ben
W^a=W^a(\mathcal{S},M)= \textrm{tr}_S[M \mathcal{S}\bullet( \mathcal{E}_{\mathbb{I}}-\Gamma)]. 
\een
It is therefore the superchannel itself that we are investigating with $W^a$, rather than the (reduced) density matrix of the system. Nevertheless, we shall find it instructive to deconstruct the superchannel into its constituent parts -- namely, the joint system environment state at $\tau$ and the joint unitary evolution $\mathcal{U}_{T,\tau}$ afterwards -- to draw our conclusions. 

In the isolated case we were able to witness the existence of coherent superpositions of classical states of the system at $\tau$: essentially by falsifying a classical view that ignores the coherences (off-diagonal elements) of the system density operator. Here we are able to test a generalisation of this idea -- namely,  to test the supposition that the coherences of the system \emph{as well as the parts of the environment that are correlated with the system coherences} can be ignored in the description of the experiment. To see this more clearly, let us write 
\begin{align}
\rho_{SE}(\tau)=\rho_S(\tau)\otimes\rho_E(\tau)+\chi_{SE}(\tau)
\label{correlation_matrix}
\end{align}
 for some correlation matrix $\chi$ and marginal states $\rho_S=\textrm{tr}_E(\rho_{SE})$ and $\rho_E=\textrm{tr}_S(\rho_{SE})$. Now,
\begin{align}
\label{Wa_contribs}
W^a(\mathcal{S},M)=& \textrm{tr}_S[M''\{\rho_S(\tau)-\Gamma(\rho_S(\tau))\}]\\
&+\textrm{tr}[(M\otimes\mathbb{I})\mathcal{U}_{T,\tau}\circ(\mathcal{E}_{\mathbb{I}}\otimes\mathcal{E}_{\mathbb{I}}-\Gamma\otimes\mathcal{E}_{\mathbb{I}})(\chi_{SE}(\tau))]\nonumber
\end{align}
Here $M''=\textrm{tr}_E(U(T,\tau)^\dagger(M\otimes\mathbb{I})U(T,\tau))$. The first term represents a contribution to the witness by coherence in the reduced state of the system, as in the isolated case. The second term represents new contributions emanating from the correlations between system and environment. The important point here is that, according to macrorealism these extra contributions would be zero. Take for example the case of a two-qubit system-environment in a Bell state:  $\rho_{SE}(\tau)=|\phi^+\rangle\langle \phi^+|=\frac{1}{2}[|0\rangle\langle0|\otimes |0\rangle\langle 0|+|1\rangle\langle1|\otimes |1\rangle\langle 1|+|0\rangle\langle1|\otimes |0\rangle\langle 1|+|1\rangle\langle0|\otimes |1\rangle\langle 0|]$. The first two terms represent classical correlations and are stabilised (unaffected) by the $\Gamma\otimes\mathcal{E}_{\mathbb{I}}$ operation, and therefore cancel from the differential witness $W^a$. The final two terms, on the other hand, are quantum correlations which are destroyed by $\Gamma\otimes\mathcal{E}_{\mathbb{I}}$. Note the reduced state of the system $\rho_S(\tau)$ is maximally mixed and therefore has no coherence (i.e. is diagonal). In such a case the reduced state of the system is no reflection of the global coherence properties of system and environment, since the quantum correlations represented by $\chi_{SE}$ can still trigger the witness due to their vulnerability to the classicalisation of the system. Hence our statement that $W^a$ reports on system coherence and/or nonclassical system-environment correlations. Again it is necessary for the measurement $M''$ to be well chosen (by the combination of $M$ and action of $U(T,\tau)$). $W^a$ therefore enables one to distinguish, with an appropriate system-environment evolution $U(T,\tau)$, proper and improper mixtures~\cite{dEspagnat2001} in the reduced state of the system at $\tau$. The former have $W^a=0$ whereas the latter can have $W^a\neq0$.  This is something which is not possible by performing local state tomography of the system at $\tau$, a procedure which only reports on properties of the reduced system state $\rho_S(\tau)=\textrm{tr}_E(\rho_{SE}(\tau))$.

A sufficient condition for the joint system environment state to give $W^a=0$ is easily seen to be membership of the set of incoherent-quantum states~\cite{YadinMaGirolami2016,AdessoBromleyCianciaruso2016,ChitambarStreltsovRana2016}, defined as
\begin{align}
\mathcal{IQ}:=\left\{\rho_{SE} : \rho_{SE} = \sum_i p_i |i\rangle\langle i | \otimes \rho_E^i\right\}.
\label{iq}
\end{align}
The set is so named since a notion of classicality (`incoherence') is imposed on the system but not on the environment part (which remains `quantum') of each term in the convex combination. Importantly, this set is a proper subset of the set of non-entangled states, meaning that its complement is a proper superset of the set of entangled states. In other words the system and environment need not be entangled, nor the reduced system state have any coherence for the witness to be triggered (a constructive example is given below). $W^a\neq0$ implies the test state $\rho_{SE}$ is not incoherent-quantum. This statement is an improvement over the findings of Li \emph{et al.}~\cite{LiLambertChen2012}, who argued that fast classicalisation should disallow false positives from the entire set of classically-correlated (non-entangled) states. Our analysis implies that some such states will in fact trigger $W^a$, identifying a class of scenarios which was missing from previous analyses.

A natural question arises: what is the maximum value of $W^a$? Maximising over the measurement for any fixed $\rho_{SE}(\tau)$ results in the trace distance between $\rho_{SE}(\tau)$ and the particular incoherent-quantum state whose marginal system state has the same diagonal entries. This is a faithful measure of distance to the set \eqref{iq}, since:
\begin{thm}
\label{thm_faithful}
The measure $\max_M W^a(\mathcal{S},M)=||\rho_{SE}-(\Gamma\otimes\mathcal{E}_{\mathbb{I}})\rho_{SE}||_{\textrm{\normalfont tr}}/2=0$ if and only if $\rho_{SE}\in\mathcal{IQ}$. %is a necessary and sufficient condition for
\end{thm}
\begin{proof}
See Appendix~\ref{ns_iq}.
\end{proof}
This measure is thus stronger than basis-dependent discord, which was shown to be a non-faithful measure due to its ascribing zero to some states outside of $\mathcal{IQ}$~\cite{YadinMaGirolami2016}. Other faithful measures exist, such as the `distillable coherence of collaboration' ~\cite{ChitambarStreltsovRana2016}.
Maximising over both measurements and states yields:
\ben
|W^a(\mathcal{S},M) |&\leq& \max_{\rho,M''} \textrm{tr}(M''[\mathcal{E}_{\mathbb{I}}\otimes\mathcal{E}_{\mathbb{I}}-\Gamma\otimes\mathcal{E}_{\mathbb{I}}]\rho)\nonumber\\
&=&\max_\rho||[\mathcal{E}_{\mathbb{I}}\otimes\mathcal{E}_{\mathbb{I}}-\Gamma\otimes\mathcal{E}_{\mathbb{I}}](\rho)||_{\textrm{tr}}/2\nonumber\\
&=:&||\mathcal{E}_{\mathbb{I}}-\Gamma||_\diamond/2
\label{dn}
\een
which is nothing other than the diamond norm distance~\cite{JohnstonKribsPaulsen2009} between the identity map and the classicalisation map. 
\begin{thm}[Diamond norm distance between identity and classicalisation channels]
\label{thm_diamond}
The diamond-norm distance between the identity channel and the classicalisation channel 
\begin{align}
||\mathcal{E}_\mathbb{I}-\Gamma||_\diamond/2 = 1-1/d.
\end{align}
where $d$ is the dimension of the Hilbert space $\mathscr{H}_S$ upon which those channels super-operate.
\end{thm}
\begin{proof}
One may find the maximum in the definition above by performing the double optimisation using the method of Lagrange multipliers, see Appendix~\ref{max_non_isolated}.
\end{proof}

This value is achievable by setting $M''=\rho_{SE}(\tau)=|+\rangle\langle+|\otimes \rho_E$, or indeed when $M''=\rho_{SE}(\tau)=|\Psi\rangle\langle \Psi|$ where $|\Psi\rangle=\sum_{i=1}^d|i\rangle|i\rangle/\sqrt{d}$ is a maximally-entangled system-environment state. Note this value is lower than would be possible if one could prepare the maximally coherent state over system and environment and dephase them both: i.e. if $M''=\rho=|+\rangle\langle +|\otimes |+\rangle\langle +|$ and if $\Gamma\otimes \mathcal{E}_{\mathbb{I}}$ became $\Gamma\otimes \Gamma$.

Theorem~\ref{thm_diamond} implies that the use of an ancilla does not help in discriminating $\Gamma$ from $\mathcal{E}_{\mathbb{I}}$ -- whereas in general there exist constructive examples of pairs of channels that may be better discriminated in such a fashion~\cite{KitaevShenVyalyi2002}: by preparing an entangled state of the enlarged system and applying one channel or the other to only part of the composite.

As a concrete example of a non-entangled state with no system coherence that can violate $W^a$, consider 
\begin{align}
\rho_{SE}(\tau)=(1-\epsilon)\mathbb{I}/d^2+\epsilon|\Psi\rangle\langle\Psi|.
\end{align}
This state has a diagonal reduced system state (i.e. no system coherence), and for $\epsilon<1/(d^2-1)$ will be non-entangled~\cite{GurvitsBarnum2002}. By linearity, however, the value of $W^a=\epsilon(1-1/d)\neq 0$ is possible by setting $M''=|\Psi\rangle\langle\Psi|$.

\subsubsection{Born approximation}
In order to witness properties of the system alone, we may consider the case where the system-environment state \emph{does} factorize at $\tau$. Such a situation is assured if an assumption known as the Born approximation (BA)~\footnote{This definition is commonplace in the open quantum systems community, and is not to be confused with distinct definitions of the Born approximation in use in scattering theory. It has a similar meaning to the Born-Oppenheimer assumption, well known to physical chemists and molecular physicists as the assumption that a molecular wavefunction can be broken into a product of electronic (`system') and nuclear (`environment') wavefunctions.}
\begin{align}
\rho_{SE}(t)=\rho_S(t)\otimes |e_0\rangle\langle e_0|\qquad \forall t,
\label{Born}
\end{align}
holds. The BA implies there are no correlations between system and environment at $\tau,$ from Eq. (\ref{correlation_matrix}) $\chi_{SE}(\tau)=0$. The BA therefore restores the inference that $W^a\neq0$ implies $ \rho_S\not\in\mathcal{I}$, $\mathcal{I}$ being the set of diagonal states. $W^a$ is thus an unambiguous witness of system coherence even for non-isolated systems, as long as the BA holds. In fact, this is true as long as $\rho_{SE}$ is in \emph{some} product state at $\tau$, with the environment not necessarily in its equilibrium state. In the next section we shall see that the full weight of the BA is necessary for $W^b$ to have such a property. %The BA further implies that the environment is in a stationary state, which is not strictly necessary for $W^a$ to witness coherence only, but will be important in the next section for $W^b$ to have such a property.

The BA may be accurate in some but not all physical situations. It is commonly employed at an intermediate step in the derivation of Master equations such as the Redfield and Quantum-Optical Master equations, routinely used to model the reduced dynamics of coupled electron-phonon~\cite{ValkunasAbramaviciusMancal2013} and atom-photon systems~\cite{BreuerPetruccione2007}. It is usually justified on the grounds of weak system-environment coupling~\cite{Haake1969,BreuerPetruccione2007}, sometimes along with appeal to the relative `largeness' of the environment~\cite{Schlosshauer2007}. It may also be motivated by the idea that the environment relaxation is sufficiently fast compared to the system dynamics, such that excitations in the environment may be neglected~\cite{KokLovett2010}. 

In reality of course, assumptions such as the BA will never be exact. Nevertheless we may quantify the accuracy of the approximation and use the quantification to temper the conclusions about coherence which may be drawn from nonzero witness values. 
\begin{thm}[]
Let $||\rho_{SE}-\textrm{\normalfont tr}_E(\rho_{SE})\otimes \textrm{\normalfont tr}_S(\rho_{SE})||_{\textrm{\normalfont tr}}=||\chi_{SE}||_{\textrm{\normalfont tr}}$ be a measure of the distance of a given system-environment state $\rho_{SE}$ to the set of product states. Then 
\begin{align}
R(\rho_S)\geq 2|W^a|-2||\chi_{SE}||_{\textrm{\normalfont tr}}.
\end{align}
\label{thm_finiteBA}
\end{thm}
\begin{proof}
Maximising the first term of Eq.~\eqref{Wa_contribs} with respect to the measurement $M''$, we have
\begin{align}
|W^a|\leq& ||\rho_S-\Gamma(\rho_S)||/2\nonumber\\
&+|\textrm{tr}[(M\otimes\mathbb{I})\mathcal{U}_{T,\tau}\circ(\mathcal{E}_{\mathbb{I}}\otimes\mathcal{E}_{\mathbb{I}}-\Gamma\otimes\mathcal{E}_{\mathbb{I}})(\chi_{SE}(\tau))]|.
\end{align}
Next, use von Neumann's trace inequality~\cite{Mirsky1975} to bound the second term:
\begin{align}
|\textrm{tr}\left((A-B)\chi_{SE}\right)|&\leq\sum_i \sigma_i(A-B)\sigma_i(\chi_{SE})\nonumber\\
&\leq \sigma_1(A-B)\sum_i\sigma_i(\chi_{SE})\nonumber\\
 &\leq||A-B||_2||\chi_{SE}||_{\textrm{tr}}\nonumber\\
 &\leq ||A||_2||\chi_{SE}||_{\textrm{tr}}\nonumber\\
 &\leq||\chi_{SE}||_{\textrm{tr}}
\end{align}
for the positive semidefinite matrices $A=\mathcal{U}^\dagger_{T,\tau}(M\otimes\mathbb{I})$ and $B=(\Gamma\otimes\mathcal{E}_{\mathbb{I}})(A)$. The notation $\sigma_i(\cdot)$ and $\lambda_i(\cdot)$ are used for singular and eigenvalues, respectively, ordered from largest to smallest. $||\cdot||_2$ is the spectral norm or largest singular value. For Hermitian matrices, the singular values are just the absolute values of the eigenvalues; the maximum singular value of $A-B$ is therefore no larger than the maximum eigenvalue of $A$ plus the maximum value of $-B$ (which is at most zero by the positive semidefiniteness of $B$). Explicitly:
\begin{align}
||A-B||_2=&\sigma_1(A-B)=|\lambda_1(A+(-B))|\nonumber\\
&\leq|\lambda_1(A)+\lambda_1(-B)|\leq |\lambda_1(A)|=||A||_2.
\end{align}
%with respective singular values $\alpha_i$ and $\beta_i$ arranged in descending order. $\gamma_i$ are the singular values of $\chi_{SE}$. 
Since $0\preceq M\otimes\mathbb{I}\preceq \mathbb{I}\otimes\mathbb{I}$ is a positive operator, and $\mathcal{U}_{T,\tau}$ is unitary transformation, $A$ has the same singular values as $M$, bounded from above by $1$. 
Substituting the bound and rearranging gives the desired result.
\end{proof}
Theorem~\ref{thm_finiteBA} enables unambiguous inference of the quantum coherence in the reduced state of the system, using $W^a$ and given an upper bound on $||\chi_{SE}||$: that is, given a quantification of the departure of the system-environment state from a product at $t=\tau$. 

Our discussion now moves on to the witness constructed in part from the second pair of interruption operations, which are considered slow compared to the relaxation of the environment. 

\subsection{$W^b$: Environment reset only during classicalisation} 
\label{subsec_Wb}
Consider the witness $W^b$ (defined in Eq.~\eqref{Wbdef} and Fig~\ref{fig1}b) that compares one interruption from each class: i.e. doing nothing $(\RomanNumeralCaps{1})$ versus classicalising the system piecewise and simultaneously resetting the environment ($\RomanNumeralCaps{4}, $ achieved on a timescale that is slow with respect to the typical environment timescales).  When the BA holds, the environment remains in $|e_0\rangle$ throughout all experiments, and $\RomanNumeralCaps{4}$ has the same effect as $\RomanNumeralCaps{2}$ interrupting the system only. 
We will now show that $W^b\neq0$ not only implies $\rho\not\in\mathcal{I}$ under the BA, but also under weaker assumptions, since we only require the environment to be in its equilibrium state at $\tau$~\cite{LiLambertChen2012} and not for all times. In fact it may even be in a distinct environment state that delivers the equivalent CP map to the system. 

Firstly, let us define
$ \mathcal{E}_{\textrm{prepare}}^{\RomanNumeralCaps{4}}$ and $\mathcal{E}_{\textrm{measure}}^{\RomanNumeralCaps{4}}$ (super-operating on $\mathscr{H}_S$) by their respective Kraus operators $K_i$ in $\mathcal{E}(\rho)=\sum_i K_i \rho K^\dagger_i$. $\mathcal{E}^\dagger$ is the dual channel, in the sense of having Kraus operators $K_i^\dagger$. Using Eq.~\eqref{interruption_equations} and the useful formula $\rho_{SE}=\rho(0)\otimes|e_0\rangle\langle e_0|=(\mathbb{I}\otimes |e_0\rangle)\rho(0)(\mathbb{I}\otimes \langle e_0|$)~\cite{pse276158}, we have

\ben
\rho_S^{\RomanNumeralCaps{4}}(T)&=&\mathcal{E}_{\textrm{measure}}^{\RomanNumeralCaps{4}}\circ\Gamma\circ\mathcal{E}_{\textrm{prepare}}^{\RomanNumeralCaps{4}}(\rho_S(0))
\een
implying
\ben
 K^{{\RomanNumeralCaps{4}},\textrm{prepare}}_i &=&\langle e_i|U(\tau,0)|e_0\rangle \nonumber\\
 K^{{\RomanNumeralCaps{4}},\textrm{measure}}_i&=& \langle e_i | U(T,\tau) | e_0\rangle. 
 \label{4mapdef}
\een
Here, $\ket{e_i}$ constitute a complete basis for $\mathscr{H}_E$.
We also choose once more to write the joint state at $\tau$ using the correlation matrix as in Eq.~\eqref{correlation_matrix}. Then the reduced environment state defines an alternative measurement CP map acting on $\mathscr{H}_S$:
\begin{align}
\mathcal{E}_{\textrm{measure}}^{\RomanNumeralCaps{1}}(\rho_S)=\textrm{tr}_E(U(T,t)[\rho_S(\tau)\otimes\rho_E(\tau)]U(T,t)).
 \label{1mapdef}
\end{align}
Using these definitions and Eqs.~\eqref{interruption_equations} and \eqref{Wbdef}, we have 
\begin{align}
W^b=& \textrm{tr}_S(M\textrm{tr}_E(U(T,\tau)\chi_{SE}U^\dagger(T,\tau))\nonumber\\
&+\textrm{tr}_S(\mathcal{E}^{\RomanNumeralCaps{1}\dagger}_{\textrm{measure}}(M)\rho_S(\tau))\nonumber\\
&-\textrm{tr}_S(\mathcal{E}^{\RomanNumeralCaps{4}\dagger}_{\textrm{measure}}(M)\Gamma(\rho_S(\tau))).
\label{Wbcontribs}
\end{align}
It is clear there are three (not mutually-exclusive) ways to have a nonzero witness value. It is necessary to have one or more of: 
\\(i) system coherence $\rho(\tau)\neq\Gamma(\rho(\tau))$,
\\(ii) system-environment correlation $\chi_{SE}\neq0$, including those \emph{inside} of $\mathcal{IQ}$, or 
\\(iii) Different CP measurement maps $\mathcal{E}_{\textrm{measure}}^{\RomanNumeralCaps{1}}\neq\mathcal{E}_{\textrm{measure}}^{\RomanNumeralCaps{4}}$. The last possibility includes the case of non-negligible excitations in the environment. 

Clearly, the BA will ensure that $W^b$ reports only on system coherence, since it forces a product structure $\chi_{SE}=0$ \emph{and} the identity of CP measurement maps $\mathcal{E}_{\textrm{measure}}^{\RomanNumeralCaps{1}}=\mathcal{E}_{\textrm{measure}}^{\RomanNumeralCaps{4}}$-- see Fig.~\ref{fig2}b. In this sense, the BA unifies $W^a,W^b$ and $W^{\textrm{isolated}}$.  

%$M'=\mathcal{E}_{\textrm{measure}}^{\RomanNumeralCaps{4}\dagger}(M)$ is a modified measurement operator. 

Li \emph{et al.} highlight a particularly worrisome failure of the BA~\cite{LiLambertChen2012}, stating that the system may ultimately have no coherence but that classical correlations between system and bath can lead to $W^b \neq 0$. This is the second of the three possibilities above. As an example, consider $U(\tau,0)=\sigma_x\otimes\sigma_x$ or $U(\tau,0)=\mathcal{E}_{\mathbb{I}}\otimes\mathcal{E}_{\mathbb{I}}$, with 50\% chance of each: a probabilistic, simultaneous excitation in system and in the environment, leading to a (potentially) only classically-correlated state, and therefore to a possibly nonzero $W^b$. We identify an additional, distinct but equally troubling scenario originating from the third point above: a classical model without any correlations can trigger a false positive. The system and environment may remain in a product but the two measurement CP maps may differ.

A specific example consists of two bits as in the right panel of Fig.~\ref{fig2}b. The environment bit undergoes a simple flip before the interruption and the system bit undergoes a conditional flip (controlled on the environment) after the interruption. So $U(\tau,0)=\mathbb{I}\otimes \sigma_x$ and $U(T,\tau) = \mathbb{I}\otimes |0\rangle\langle 0| + \sigma_x\otimes |1\rangle\langle 1|$. This is a classically controlled, conditional unitary evolution. Nevertheless, when $M=|0\rangle\langle0|$ it results in $W^b$ taking the maximum algebraic value of $1$, as is easily verified. The reason is simply that the interruption caused the environment to reset, which then reset the subsequent CP map delivered to the system and therefore reset the effective measurement operator. If we were to call the state of the environment a `hidden variable', then our model would closely resemble Montina's time-correlated noise model of a qubit~\cite{Montina2012}. To decide whether our example dynamics is Markovian we would first need to fix the definition of Markovianity -- some authors~\cite{Montina2012, SmirneEgloffDiaz2018} take time-inhomogenous evolutions, which include an explicit dependence on time such as our example here, to be non-Markovian. The analysis of Ref.~\cite{PollockRodriguez-RosarioFrauenheim2018} would also class it as non-Markovian, given that the state of the system at $T$ would depend on the choice of interruption (or `control') operation at $\tau$ even if a causal break were introduced just afterwards.

Interpreting the third possibility as a loophole, it may be narrowed by taking $T-\tau$ very small. Then, $\mathcal{U}(T,\tau)\approx\mathcal{E}_{\mathbb{I}}$, meaning $\mathcal{E}_{\textrm{measure}}^{\RomanNumeralCaps{1}}\approx\mathcal{E}_{\textrm{measure}}^{\RomanNumeralCaps{4}}\approx\mathcal{E}_{\mathbb{I}}$. In such a case, $M'\approx M$ meaning $W^b$ is suppressed (regardless of $\rho_{SE}(\tau)$) if $M$ projects onto a diagonal state.

These examples are subsumed by the following theorem, which enables unambiguous inference of the quantum coherence in the reduced state of the system, using $W^b$ and given an upper bound on $||\chi_{SE}||_{\textrm{tr}}$ \emph{and} an upper bound on the distance between $\mathcal{E}^{\RomanNumeralCaps{1}\dagger}_{\textrm{measure}}$ and $\mathcal{E}^{\RomanNumeralCaps{4}\dagger}_{\textrm{measure}}$. That is, given a quantification of the departure of the system-environment state at $t=\tau$ from the BA class defined in~\eqref{Born}.
\begin{prop}
\label{prop1}
Let $|||\mathcal{E}_{\textrm{\normalfont measure}}^{\RomanNumeralCaps{1}}-\mathcal{E}_{\textrm{\normalfont measure}}^{\RomanNumeralCaps{4}}|||_{\textrm{\normalfont tr}}:=\max_\rho||(\mathcal{E}_{\textrm{\normalfont measure}}^{\RomanNumeralCaps{1}}-\mathcal{E}_{\textrm{\normalfont measure}}^{\RomanNumeralCaps{4}})\rho||_{\textrm{\normalfont tr}}$ 
be the induced (superoperator) trace norm distance between the two measurement maps in $W^b$.
This distance is bounded by the trace distance between the reduced and thermal equilibrium states of the enrvironment:
\begin{align}
|||\mathcal{E}_{\textrm{\normalfont measure}}^{\RomanNumeralCaps{1}}-\mathcal{E}_{\textrm{\normalfont measure}}^{\RomanNumeralCaps{4}}|||_{\textrm{\normalfont tr}}\leq ||\textrm{\normalfont tr}_S(\rho_{SE})-\ket{e_0}\bra{e_0}||_{\textrm{\normalfont tr}}.
\end{align}
%: $|||\cdot|||_{\textrm{tr}}$ denotes the induced trace norm for superoperators.
\end{prop}
\begin{proof}
Using the definitions in Eqs.~\eqref{4mapdef} and~\eqref{1mapdef}, we have:
\begin{align}
\max_{\mathcal{U}_{T,\tau}}|||\mathcal{E}_{\textrm{measure}}^{\RomanNumeralCaps{1}}-\mathcal{E}_{\textrm{measure}}^{\RomanNumeralCaps{4}}|||_{\textrm{tr}}=\nonumber\\
\max_{\mathcal{U}_{T,\tau},M,\rho_S}\textrm{tr}[(M\otimes\mathbb{I})\mathcal{U}_{T,\tau}(\rho_S\otimes[\textrm{tr}_S(\rho_{SE})-|e_0\rangle\langle e_0|])]\nonumber\\
=\max_{\rho_S}||\rho_S\otimes[\textrm{tr}_S(\rho_{SE})-|e_0\rangle\langle e_0|]||_{\textrm{tr}}\nonumber\\
=\max_{\rho_S}||\rho_S||_{\textrm{tr}}||\textrm{tr}_S(\rho_{SE})-|e_0\rangle\langle e_0|||_{\textrm{tr}}\nonumber\\
=||\textrm{tr}_S(\rho_{SE})-|e_0\rangle\langle e_0|||_{\textrm{tr}},
\end{align}
using the multiplicativity of the trace norm with respect to tensor products, and the unit trace-norm of density matrices.
\end{proof}
\begin{thm}[]
As above, let $||\rho_{SE}-\textrm{\normalfont tr}_E(\rho_{SE})\otimes \textrm{\normalfont tr}_S(\rho_{SE})||_{\textrm{tr}}=||\chi_{SE}||_{\textrm{tr}}$ be a measure of the distance of a given system-environment state $\rho_{SE}$ to the set of product states. Then 
\begin{align}
R(\rho_S)\geq 2|W^b|-2||\chi_{SE}||_{\normalfont\textrm{tr}}-2||\textrm{\normalfont tr}_S(\rho_{SE})-\ket{e_0}\bra{e_0}||_{\textrm{\normalfont tr}}.
\end{align}
\label{thm_finiteBA2}
\end{thm}
\begin{proof}
Begin by maximizing the absolute value of~\eqref{Wbcontribs}:
\begin{align}
|W^b|\leq& \max_{M,U(T,\tau)}  \textrm{tr}([M\otimes\mathbb{I}](U(T,\tau)\chi_{SE}U^\dagger(T,\tau))\nonumber\\
&+\max_{M,\mathcal{E}^{\RomanNumeralCaps{4}\dagger}_{\textrm{measure}}}\textrm{tr}_S(\mathcal{E}^{\RomanNumeralCaps{4}\dagger}_{\textrm{measure}}(M)[\rho_S-\Gamma(\rho_S)])\nonumber\\
&+\max_{M,\rho_S}\textrm{tr}_S([\mathcal{E}^{\RomanNumeralCaps{1}\dagger}_{\textrm{measure}}-\mathcal{E}^{\RomanNumeralCaps{4}\dagger}_{\textrm{measure}}](M)(\rho_S).\nonumber \\
\leq& ||\chi_{SE}||_{\textrm{tr}}+R(\rho_S)/2 + |||\mathcal{E}^{\RomanNumeralCaps{1}\dagger}_{\textrm{measure}}-\mathcal{E}^{\RomanNumeralCaps{4}\dagger}_{\textrm{measure}}|||_{\textrm{tr}}.
\end{align}
We used the identity $\mathcal{E}^{\RomanNumeralCaps{1}\dagger}=\mathcal{E}^{\RomanNumeralCaps{1}\dagger}-\mathcal{E}^{\RomanNumeralCaps{4}\dagger}+\mathcal{E}^{\RomanNumeralCaps{4}\dagger}$. Replacing the third term using Proposition~\ref{prop1} and rearranging gives the desired result.
\end{proof}

Since we assume the system and environment to begin in a product, there cannot be any correlations at $\tau$ unless the environment undergoes some nontrivial dynamics. Hence our conclusion that $W^b$ reports on system coherence and/or coupling to a non-stationary environment. 
\subsection{$W^c$: Environment reset during both interruptions}
\label{subsec_Wc}
If it were possible to artificially reset the environment at $\tau$, this would remove the loopholes detailed in the previous section and in the right panel of Fig.\ref{fig2}b. Our fix to the loopholes requires an increase in the number of experiments, and involves replacing the completely uninterrupted experiment $(\RomanNumeralCaps{1})$ with one that reconstructs the full reduced state at $\tau$, namely interruption $(\RomanNumeralCaps{3})$. This removes the need to assume the BA in Eq.~\eqref{Born}, but requires order $d^2$ experiments that are capable of inferring (and re-preparing) $\rho$ in full, including its coherences.

This solution is equivalent to actively intervening into the dynamics of the environment and enforcing validity of the BA -- see Fig.~\ref{fig2}c. Now both interruptions in Eqs.~(\ref{eq:int3},\ref{eq:int4}) reset the environment and we have 
\ben
W^c:=P^{\RomanNumeralCaps{3}}-P^{\RomanNumeralCaps{4}} = \textrm{tr} [M' (\rho(\tau)-\Gamma(\rho(\tau))],
\een
which shows that $W^c$ reports only on the coherence in the reduced state $\rho(\tau)$. It is also clear that the false positive for $W^b$ that we outline in the right panel of Fig~\ref{fig2}b is not a false positive for $W^c$, yielding $W^c=0$.

The downside to this approach is that one has implicitly performed state tomography at $\tau$, and any desired witness or measure may be calculated directly, making the remainder of the protocol superfluous. %Alternative solutions include decoupling the environment in real time. 

\section{Device independence}
\label{sec_device}
Is it necessary to trust $\Gamma$ in order to draw a conclusion about the nonclassicality of a system via these witnesses? This question takes on increased significance given our focus on non-isolated systems where perfect implementations of all operations (including $\Gamma$) may not be available. In particular, in realistic experiments it will almost certainly be the case that 
\begin{align}
W(\rho_c)\neq0
\end{align}
despite the converse being predicted by quantum theory, due to statistical or other types of noise. It would therefore arguably be premature to ascribe nonzero witness values to unambiguous quantumness in the system state. There is a solution to this conundrum, however. 

By running the isolated witness test for the initial state set to each of the classical states $|i\rangle$ in turn (say with appropriate control of $U(\tau,0)$), one can set a baseline for the witness $W_i$ for each of them. Then, allowing for the preparation of a test state, we can test to see if the witness exceeds the range of baseline values. Such an approach was used in~\cite{KneeKakuyanagiYeh2016,WangKneeZhan2017}, and tests if
\begin{align}
W^{\textrm{isolated}}\in [\min(W_i),\max(W_i)].
\label{baseline_interval}
\end{align}
By the linearity of the Born rule, this condition is satisfied if the test state is merely some (arbitrarily weighted) convex mixture of the classical states (i.e. is diagonal). Hence violation is sufficient to infer that the test state is not diagonal without the need to trust the action of $\Gamma$. 

The idea generalises from the isolated case to our new witnesses $W^{a,b,c}$, in the sense that the condition will be satisfied for any arbitrary convex mixture of the joint system-environment states $\rho_{SE}(\tau)_i$ used to define the baseline interval $[\min(W_i),\max(W_i)]$. Then one can rely on the promise that each of these baseline states belonged to a certain convex set (e.g. incoherent-quantum $\mathcal{IQ}$) to conclude that a violation of Eq.~\eqref{baseline_interval} witnesses the non-inclusion of the test state in that set. Our proposal, which extends the isolated prototype Eq.~\eqref{baseline_interval}, therefore trades trust in $\Gamma$ for some level of trust in preparing the baseline states, e.g.  $\{|i\rangle\}_i$, lending it a quasi-device-independent property.

Clearly, one needs to be sure that $W_i$ constitute an exhaustive characterisation of the extreme points of the set -- otherwise violation of Eq.~\eqref{baseline_interval} can be caused by the test state living in an uncharacterised subspace. It could be that the system is higher dimensional than expected, for example, or that the test state is `more pure' than any of the baseline states. In other words, Eq.~\eqref{baseline_interval} is simply testing whether the test state is in the convex hull of the baseline states. Alternative readings of macrorealism known as `eigenstate support' and `supra eigenstate support' given by Maroney and Timpson~\cite{MaroneyTimpson2014} rely on precisely these ideas to maintain a classical view in the face of apparent violations of macrorealist conditions. In these alternative readings of macrorealism, `hidden variables' play a non-trivial role because the true classical (or `ontic') states can no longer be fully described from within the quantum formalism, i.e. merely as diagonal density operators on a space of known dimension.

Building on the idea that Eq.~\eqref{baseline_interval} defines a quasi-device-independent test of eigenstate-mixture macrorealism, generalisations of $\Gamma$ to arbitrary CPTP maps have been considered~\cite{WangKneeZhan2017,MoreiraCunha2018}. Replacing $\Gamma$ with an arbitrary channel $\mathcal{E}$, we have 
\begin{align}
V_{\mathcal{E}}(\rho,M):= \textrm{tr}(M[\rho-\mathcal{E}(\rho)])\leq ||\rho-\mathcal{E}(\rho)||_{\textrm{tr}}.
\end{align}
It is possible to find an $\mathcal{E}$ such that each state in the preferred basis is unaltered but that coherent superpositions are taken to orthogonal states~\cite{WangKneeZhan2017}, thus giving the maximum algebraic violation of~\eqref{baseline_interval}.
These ideas have connections to resource theories of asymmetry~\cite{MarvianSpekkensZanardi2016}.

\section{Conclusion}
\label{sec_conclusion}
We introduced three new witnesses in Eqs.~\eqref{Wadef}, \eqref{Wbdef}, and \eqref{Wcdef}, generalising the No-Signalling-In-Time witness of quantum coherence to apply to non-isolated systems. Prescribing the `best' generalised witness depends on the application. $W^a$ is relatively cheap but requires specific quantum control that may or may not be available. The specific quantum control in question, $\Gamma\otimes\mathcal{E}_{\mathbb{I}}$, is an operation commonly thought to be cheap to perform in the laboratory: It greatly reduces the experimental cost of inferring quantum coherence (with respect to state tomography). $W^a$ is not solely a witness of quantum coherence in the system of interest, but tests for membership of $\mathcal{IQ}$: the `incoherent-quantum' set of joint system-environment states. The latter is arguably the more interesting and meaningful notion of classicality, and the former is recovered as a special case when the system environment state factorizes (ensured when the Born approximation Eq.~\eqref{Born} holds).

 $W^b$ is more expensive, yet still cheaper than tomography, but relies on the validity of additional assumptions (again, ensured by the Born approximation) to rule out violations due to excitations in the environment. $W^c$ does not rely on such an approximation, but is expensive, and one may as well use the implicit state tomogram generated during the protocol to calculate any desired property,  including for example the resource-theoretic measure $R(\rho)$. 
 
 In summary, we have shown:  
 
 \begin{itemize}
 \item There is a link between tests of macrorealism and a certain resource-theoretic coherence monotone $R(\rho)$ (defined in Eq.~\eqref{vulnerability_of_coherence}) through the NSIT protocol and associated witness (Section~\ref{sec_isolated}). 
 \item Admitting the existence of an environment leads to the realisation that a specific class of system-environment correlations (weaker than entanglement and weaker than discord) can trigger $W^a$ when there is no quantum coherence in the system. We proved the maximum violation of the witness is unchanged from the isolated case, and give two explicit examples achieving it in arbitrary dimension (Section~\ref{subsec_Wa}).
 \item Performing NSIT with slow classicalisation has a wider loophole than previously believed: the system can be classical and entirely uncorrelated with a non-stationary environment, and this can trigger violations. We gave an explicit example achieving the maximum algebraic violation (Section~\ref{subsec_Wb}).
 \item Enforcing a stationary environment can alleviate the loophole by enforcing the Born approximation but is as expensive as full quantum state tomography (Section~\ref{subsec_Wc}).
 \end{itemize}

These results should assist the design, execution and interpretation of anticipated tests of quantumness in biological and other non-isolated systems. 
The classicalisation operation $\Gamma$ has a counterpart in the theory of generalised probabilistic theories (GPTs)~\cite{Barrett2007} -- see for example~\cite{ScandoloSalazarKorbicz2018} where the counterpart is named `complete decoherence' -- future work may investigate the validity of the notion of coherence as a resource in GPTs.

\begin{acknowledgments}
G.C.K. was supported by the Royal Commission for the Exhibition of 1851, and thanks Neill Lambert and Benjamin Yadin for helpful discussions. AD and MM are supported by UK EPSRC (EP/K04057X/2). 
\end{acknowledgments}

\begin{appendix}

%%%%%%%%%%%%%%%%%%%
\section{Partial summation method}
%%%%%%%%%%%%%%%%%%%
\label{partial_sum}
Li \emph{et al.} propose a method to potentially reduce the number of experiments required to witness coherence~\cite{LiLambertChen2012}. With a `piecewise' implementation of $\Gamma$ in mind (where the system is measured and conditionally re-prepared in each of the $d$ possible classical states at $\tau$), they write 
\be
|W^{\textrm{isolated}}(\rho,M')|=|p_m(T)-\sum_n p_n(\tau) \Omega_{mn}(T,\tau)|
\ee
where $\Omega_{mn}(T,\tau)$ are the conditional probabilities that the system will be found in state $n$ at $T$ given that it was found in state $m$ at $\tau$. Li \emph{et al.} state that if the terms in the sum are found in separate experiments performed sequentially, such a procedure may be stopped early, `as soon as the witness is violated by this partial summation'. This is not an entirely safe prescription, however, since e.g. for a qubit if $M'=|+\rangle\langle +|,\rho=|1\rangle\langle 1|$ the witness is satisfied (takes a zero value) only when the complete sum is constructed. Therefore stopping early could lead to a false positive. The refined statement from Li \emph{et al.}, that the experiments can be stopped as soon as the terms in the sum together are larger than $p_m(T)$, on the other hand, \emph{is} a safe prescription because the sum is monotonically increasing with the number of terms while $p_m(T)$ is constant. However, some violations of $W^{\textrm{isolated}}(\rho,M')=0$ involve the sum always remaining \emph{below} $p_m(T)$. Take for example $M'=\rho=|+\rangle\langle + |$. In this case there is apparently no cost saving to be had through partial summation. However one may simply use the complementary measurement operator $M'\rightarrow \mathbb{I}-M'$ to define a new witness for the same experiment which \emph{does} have the desirable partial summation property. The prescription we suggest is that if $p_m(T)>\frac{1}{2}$, we may define a smaller $q_m(T):=1-p_m(T)$ from the same experimental data such that the monotonically growing partial summation term (whose summands are now also the complements of the previous values) crosses the new constant value sooner. When the state and measurement are chosen to maximise the value of the witness, the terms in the sum are all equal~\cite{SchildEmary2015} and the order of partial summation does not matter. 

%%%%%%%%%%%%%%%%%%%
\section{Maximum NSIT violation for isolated systems}
\label{max_isolated}
\begin{thm_isolated_repeated}[Isolated case]
The induced trace-norm distance between the identity channel $\mathcal{E}_{\mathbb{I}}$ and the classicalisation channel $\Gamma$ is $1-1/d$ where $d$ is the dimension of the Hilbert space upon which those channels super-operate.
\end{thm_isolated_repeated}
\begin{proof}
The definition of the induced trace norm distance between $\mathcal{E}_{\mathbb{I}}$ and $\Gamma$ channels is 
\begin{align}
&\max_{\rho\succeq0,\textrm{Tr}\rho=1} ||\mathcal{E}_{\mathbb{I}}(\rho)-\Gamma(\rho)||_{\textrm{tr}}/2\nonumber\\&= 
\max_{\rho\succeq0,\textrm{Tr}\rho=1} \max_{M\succeq0} \textrm{Tr}[M(\rho-\Gamma(\rho))]\nonumber\\
&=\max_{\braket{\psi|\psi}=1} \max_{M\succeq0} \textrm{Tr}[M(\ket{\psi}\bra{\psi}-\Gamma(\ket{\psi}\bra{\psi}))]\label{restrict_to_pure}\\
&= \max_{\braket{\psi|\psi}=1} \max_{M^2=M} \textrm{Tr}[M(\ket{\psi}\bra{\psi}-\Gamma(\ket{\psi}\bra{\psi}))]\label{restrict_to_projectors}\\
&= \max_{\braket{\psi|\psi}=1} \max_{\braket{\phi^a|\phi^b}=\delta_{ab}}\sum_{a=1}^{\textrm{rank}(M)}\bra{\phi^a}(\ket{\psi}\bra{\psi}-\Gamma(\ket{\psi}\bra{\psi})\ket{\phi^a}.
\end{align}
Eq.~\eqref{restrict_to_pure} ($\rho\rightarrow|\psi\rangle\langle\psi|$) is justified since the maximum of a linear function over a convex set (i.e. the set of density matrices) is always achieved on the boundary of that set (i.e. on a pure state). In Eq.~\eqref{restrict_to_projectors}, we used the fact that $M$ may be taken as a projective operator~\cite{NielsenChuang2004,Fuchs1996}.
In the last step, we used the spectral theorem to write $M=\sum_a \ket{\phi^a}\bra{\phi^a}$, where the eigenvalues of $M$ are either $0$ or $1$ by its projective property $M^2=M$. 

The Lagrangian of this constrained optimisation problem is
\begin{align}
\mathcal{L}=&\sum_{a=1}^{\textrm{rank}(M)}\bra{\phi^a}(\ket{\psi}\bra{\psi}-\Gamma(\ket{\psi}\bra{\psi}))\ket{\phi^a}\nonumber\\&+\lambda_\psi\left(1-\braket{\psi|\psi}\right)+\sum_{a=1}^{\textrm{rank}(M)}\sum_{b=1}^{\textrm{rank}(M)}\lambda^{a,b}_\phi\left(\delta_{ab}-\braket{\phi^b|\phi^a}\right)\nonumber\\
=&\sum_{a=1}^{\textrm{rank}(M)}\sum_{i}\sum_{j\neq i}\bar{\phi}_i^a\phi_j^a\psi_i\bar{\psi}_j+\lambda_\psi\left(1-\sum_i \psi_i\bar{\psi}_i\right)\nonumber\\&+\sum_{a=1}^{\textrm{rank}(M)}\sum_{b=1}^{\textrm{rank}(M)}\lambda_\phi^{a,b}\left(\delta_{ab}-\sum_i \phi^a_i\bar{\phi}^b_i\right).
\end{align}
Recall that $\Gamma(\rho)=\sum_i \ket{i}\bra{i}\rho\ket{i}\bra{i}$. We expanded all operators in the preferred basis $\{\ket{i}\}_{i=1}^d$, which makes $\ket{\psi}\bra{\psi}-\Gamma(\ket{\psi}\bra{\psi})$ a `hollow' matrix --- meaning that its diagonal entries are zero. Complex conjugates are denoted with a bar, and are treated as independent variables for the purposes of differentiation, i.e. $\partial z/\partial \bar{z}=\partial \bar{z}/\partial z=0$. Setting the derivatives with respect to all parameters and with respect to the Lagrange multipliers $\lambda_\psi,\lambda_\phi^{a,b}$ equal to zero enforces the constraints and yields conditions for optimality of $\psi$ and $\phi$:
\begin{align}
\bar{\psi}_m&=\sum_{a=1}^{\textrm{rank}(M)}\frac{\bar{\phi}_m^a}{\lambda_\psi}\sum_{i\neq m}\phi_i^a\bar{\psi}_i\label{psi_bar}\\
\psi_m&=\frac{\sum_{b=1}^{\textrm{rank}(M)}\lambda_\phi^{a,b}\phi_m^b}{\sum_{i\neq m}\phi_i^a\bar{\psi}_i}.
\label{psi}
\end{align}
Multiplying these conditions together gives
\begin{align}
|\psi_m|^2 = \sum_a^{\textrm{rank}(M)} \sum_b^{\textrm{rank}(M)} \bar{\phi}_m^a\phi_m^b \frac{\lambda_\phi^{a,b}}{\lambda_\psi}\label{prod},
\end{align}
summing over $m$, gives 
\begin{align}
\lambda_\psi = \sum_a^{\textrm{rank}(M)} \lambda_\phi^{a,a},
\end{align}
by the orthonormality relation $\sum_m \bar{\phi}_m^a \phi_m^b=\delta_{ab}$ for $\delta_{ab}$ the Kronecker delta. Then, multiplying Eq.~\eqref{psi} by $\bar{\phi}^c_m\sum_{i\neq m}\phi_i^a \bar{\psi}_i$ and summing over $m$ gives 
\begin{align}
\lambda_\phi^{a,c} = \sum_m \bar{\phi}^c_m\psi_m\sum_{i\neq m}\phi_i^a \bar{\psi}_i,
\label{lam}
\end{align}
using the same orthonormality relation. Substituting Eq.~\eqref{lam} into~\eqref{psi} gives 
\begin{align}
\psi_m = \sum_c^{\textrm{rank}(M)} \sum_m \bar{\phi}_m^c\phi_m^c \psi_m.
\end{align}
Since at least one $\psi_m$ must be non-zero, we may divide by it:
\begin{align}
1 = \sum_c^{\textrm{rank}(M)}\sum_m \bar{\phi}_m^c\phi_m^c=\textrm{rank}(M).
\end{align}
This collapses the sums over $a$ and $b$ in Eqs.~\eqref{prod} and also refines~\eqref{lam} to $\lambda_\psi=\lambda_\phi^{1,1}$, giving
\begin{align}
|\psi_m| = |\phi^1_m|.
\end{align}
Substituting the condition into the objective function (eliminating the measurement which we have shown must be a rank one projector), we have:
\begin{align}
&\max_{\braket{\psi|\psi}=1}\sum_i\sum_{j\neq i }|\psi_i|^2|\psi_j|^2e^{i(\theta_i+\varphi_i-\theta_j-\varphi_j)}\nonumber\\
&=\max_{|\psi_i|}\sum_i|\psi_i|^2(1-|\psi_i|^2)\nonumber\\
&=\max_{|\psi_i|}\left(1-\sum_i|\psi_i|^4\right).
\end{align}
This is a necessary condition for maxima. We wrote the expansion coefficients in polar form $\psi_i=|\psi_i|e^{i\theta_i},\phi_i^a=|\phi_i^a|e^{i\varphi_i}$ and replaced the overall phase of each term in the sum with unity, resulting in an achievable upper bound.  Maximising over $\psi$ again yields $|\psi_i|=1/\sqrt{d}$, giving
\begin{align}
\max_{\rho\succeq0,\textrm{Tr}\rho=1} ||\mathcal{E}_{\mathbb{I}}(\rho)-\Gamma(\rho)||_{\textrm{tr}}/2=1-1/d.
\end{align}
\end{proof}
%%%%%%%%%%%%%%%%%%%
%%%%%%%%%%%%%%%%%%%
\section{Mixed states of the environment}
\label{mixed_env}
In the main text we made a zero temperature assumption for the sake of brevity. In this section we relax that assumption and show that arguments remain essentially unchanged. Let the joint state be 
\begin{align}
\rho_{SE}(t_i)=\rho_S(t_i)\otimes\rho_E^{th.eq.}
\end{align}
for $\rho_E^{th.eq.}=\sum_i p_i |e_i\rangle\langle e_i|$ being the thermal equilibrium state of the environment at finite temperature, expanded here in the energy eigenbasis. $p_k$ are commonly taken to be a Boltzmann distribution. Now the reduced dynamics of the system can be written 
\begin{align}
\rho_S(t_j) &= \textrm{tr}_E(U(t_j,t_i)\rho_{SE}(t_i)U^\dagger(t_j,t_i))\nonumber\\
&=\sum_{ik} (\sqrt{p_i}\langle e_k|U(t_j,t_i)|e_i\rangle)\rho_{S}(t_i) (\sqrt{p_i}\langle e_i| U(t_j,t_i)^\dagger|e_k\rangle)\nonumber\\
&=\sum_{ik} K_{ik}\rho_S(t_i)K_{ik}^\dagger\nonumber\\
&=:\mathcal{E}(\rho_S(t_i))
\end{align}
where $K_{ik}=\sqrt{p_i}\langle e_k|U(t_j,t_i)|e_i\rangle$ are Kraus operators satisfying $\sum_{ik}K_{ik}K^\dagger_{ik}=\mathbb{I}$. Our analysis can then be re-run with this more general definition of the Kraus operators defining a CP map.%%%%%%%%%%%%%%%%%%%
%%%%%%%%%%%%%%%%%%%
\section{Proof of faithfullness of trace distance for incoherent-quantum states}
%%%%%%%%%%%%%%%%%%%
\label{ns_iq}
\begin{thm_faithful_repeated}
 $||\rho_{SE}-(\Gamma\otimes\mathcal{E}_{\mathbb{I}})\rho_{SE}||/2 =  0 $ if and only if $\rho_{SE} \in \mathcal{IQ}$.
\end{thm_faithful_repeated}
\begin{proof}
For the forward direction, apply $(\Gamma\otimes\mathcal{E}_{\mathbb{I}})$ to $\rho$ using the definition of $\mathcal{IQ}$ from Eq.~\eqref{iq}:
\begin{align}
(\Gamma\otimes\mathcal{E}_{\mathbb{I}})\rho_{SE}&=(\Gamma\otimes\mathcal{E}_{\mathbb{I}})\sum_i p_i |i\rangle\langle i | \otimes \rho_E^i\nonumber\\
&=\sum_k\sum_ip_i \braket{k|i}\braket{i|k}\otimes\rho_E^i\nonumber\\
&=\sum_k p_k |k\rangle\langle k | \otimes \rho^k_E=\rho_{SE}.
\end{align}
For the reverse direction, the fact that $||\cdot||_\textrm{tr}$ is a norm implies that $\rho_{SE}=(\Gamma\otimes\mathcal{E}_{\mathbb{I}})\rho_{SE}$. Next, spectrally decompose $\rho_{SE}$
\begin{align}
\rho_{SE}=(\Gamma\otimes\mathcal{E}_{\mathbb{I}})\rho_{SE} = (\Gamma\otimes\mathcal{E}_{\mathbb{I}})\sum_k P_k |\psi^k\rangle\langle\psi^k|,
\end{align}
and write each pure state in the convex combination in Schmidt form~\cite{NielsenChuang2004} $|\psi^k\rangle = \sum_i\lambda_i^k|\phi_i^k\rangle\otimes|\chi^k_i\rangle, \text{where}~ \lambda^k_i\geq0, \sum_i(\lambda^k_i)^2=1, \braket{\phi^k_i|\phi^k_j} = \braket{\chi^k_i|\chi^k_j}=\delta_{ij}$:
\begin{align}
\rho_{SE}&=(\Gamma\otimes\mathcal{E}_{\mathbb{I}})\sum_{kij}P_k\lambda_i^k\lambda_j^k |\phi_i^k\rangle\langle\phi_j^k|\otimes|\chi_i^k\rangle\langle\chi_j^k|\nonumber\\
&=\sum_{lkij}P_k\lambda_i^k\lambda_j^k |l\rangle\langle l|\phi_i^k\rangle\langle\phi_j^k|l\rangle\langle l|\otimes|\chi_i^k\rangle\langle\chi_j^k|\nonumber\\
&= \sum_l p_l |l\rangle\langle l | \otimes \rho_E^l \in \mathcal{IQ},
\end{align}
where we defined 
\begin{align}
\rho_E^l &= \frac{1}{p_l}\sum_{k}p_k\left[\sum_i \lambda_i^k \braket{l|\phi_i^k}|\chi_i^k\rangle\right]\left[\sum_j \lambda_j^k \braket{\phi_j^k|l}\bra{\chi_j^k}\right]\\
p_l&=\textrm{tr}\left(\sum_{k}p_k\left[\sum_i \lambda_i^k \braket{l|\phi_i^k}|\chi_i^k\rangle\right]\left[\sum_j \lambda_j^k \braket{\phi_j^k|l}\bra{\chi_j^k}\right]\right)\nonumber\\
&=\sum_{mkij}P_k\lambda_i^k\lambda_j^k\braket{l|\phi_i^k}\braket{\chi_m^k|\chi_i^k}\braket{\chi_j^k|\chi_m^k}\braket{\phi_j^k|l}\nonumber\\
&=\sum_{km}P_k (\lambda^k_m)^2|\braket{l|\phi_m^k}|^2.
\end{align}
The proof goes through if $\rho_E^l$ are density operators, and if $p_l\geq0$ and form a resolution to unity. From the definitions, $\rho_E^l$ are manifestly positive semidefinite, and we have 
\begin{align}
\sum_l p_l &=\sum_{lmk}P_k\lambda_m^2\langle \phi_j^k|l\rangle\langle l |\phi_j^k\rangle\nonumber\\
&=\sum_{mk}P_k\lambda_m^2\langle \phi_j^k|\phi_j^k\rangle=\sum_kP_k\sum_m\lambda_m^2=1.
\end{align}
%This completes the proof.
\end{proof}

%%%%%%%%%%%%%%%%%%%
\section{Maximum violation for non-isolated case}
\label{max_non_isolated}
%%%%%%%%%%%%%%%%%%%
\begin{thm_diamond_repeated}[Non-isolated case]
The diamond-norm distance between the identity channel $\mathcal{E}_{\mathbb{I}}$ and the classicalisation channel $\Gamma$ is equal to $1-1/d$ where $d$ is the dimension of the Hilbert space upon which those channels super-operate.
\end{thm_diamond_repeated}
\begin{proof}
\begin{align}
||\mathcal{E}_\mathbb{I}-\Gamma||_\diamond/2 :& = \max_{\rho\succeq0,\textrm{Tr}\rho=1} ||\rho-[\Gamma\otimes\mathcal{E}_{\mathbb{I}}](\rho)||_{\textrm{tr}}.
\end{align}
Here $\rho=|\psi\rangle\langle\psi|$ acts on $\mathscr{H}_S\otimes\mathscr{H}_E$, having dimension $d^2$. In principle the maximisation should allow for an environment of arbitrary dimension: however, since the diamond norm is proven to be \emph{stable}, we need not consider environments of dimension greater than that of the system~\cite{Kitaev1997,GilchristLangfordNielsen2005}. We will find it convenient to expand in a system environment basis $|i,\alpha\rangle$ where $i$ enumerates the classical preferred basis of the system (as above) and $\alpha$ enumerates some basis of the environment. Each index runs over $d$ values. Then, $\mathcal{H}=\rho-[\Gamma\otimes\mathcal{E}_{\mathbb{I}}](\rho)$ is `block-hollow': $\mathcal{H}_{i\alpha,j\beta}=\mathcal{H}_{i\alpha,j\beta}(1-\delta_{ij})$. Following the proof of Theorem~~\ref{thm_isolated} (shown above in Appendix~\ref{max_isolated}). Beginning from Eq.~\eqref{dn}.
\begin{align}
||\mathcal{E}_\mathbb{I}-\Gamma||_\diamond/2 &=\max_{\psi_{i\alpha},\phi_{i\alpha}} \sum_{a=1}^{\textrm{rank}(M'')} \sum_{i,\alpha}\sum_{j\neq i,\beta} \bar{\phi}^a_{i\alpha}\phi^a_{j\beta}\psi_{i\alpha}\bar{\psi}_{j\beta},
\end{align}
subject to corresponding orthonormality and normalisation constraints on $\rho=|\psi\rangle\langle\psi|,\,|\psi\rangle=\sum_{i,\alpha}\psi_{i\alpha}|i,\,\alpha\rangle,\,M''=\sum_a^{\textrm{rank}(M'')}|\phi^a\rangle\langle \phi^a|,\, |\phi^a\rangle = \sum_{i,\alpha} \phi_{i\alpha}^a|i,\alpha\rangle$.
The method of Lagrange multipliers yields
\begin{align}
\bar{\psi}_{m\gamma}&=\sum_{a=1}^{\textrm{rank}(M'')}\frac{\bar{\phi}_{m\gamma}^a}{\lambda_\psi}\sum_{i\neq m,\beta}\phi_{i\beta}^a\bar{\psi}_{i\beta}\\
\psi_{m\gamma}^a&=\frac{\sum_b^{\textrm{rank}(M'')} \lambda_\phi^{a,b}\phi_{m\gamma}^b}{\sum_{i\neq m,\alpha }\phi_{i\alpha}\bar{\psi}_{i\alpha}}.\label{psi_ni}
\end{align}
Multiplying the conditions together gives 
\begin{align}
|\psi_{m\beta}|^2=\sum_a^{\textrm{rank}(M'')} \sum_b^{\textrm{rank}(M'')} \frac{\lambda_\phi^{a,b}}{\lambda_\psi}\phi_{m\beta}^b \bar{\phi}_{m\beta}^a.
\label{prod_ni}
\end{align}
Summing over $m$ and $\beta$ gives
\begin{align}
\lambda_\psi=\sum_a^{\textrm{rank}(M'')} \lambda_\phi^{a,a}.
\label{lambda_ni}
\end{align}
Next, multiply Eq.~\eqref{psi_ni} by $\bar{\phi}_{m\beta}^c\sum_{i\neq m,\alpha}\phi_{i\alpha}^a\bar{\psi}_{i\alpha}$, and sum over $m$ and $\beta$:
\begin{align}
\lambda_\phi^{a,c} = \sum_{m,\beta}\psi_{m\beta}\bar{\phi}_{m\beta}^c\sum_{i\neq m,\alpha}\phi_{i\alpha}^a\bar{\psi}_{i\alpha}.
\end{align}
Substituting back into Eq.~\eqref{psi_ni} yields
\begin{align}
\psi_{m\beta}&= \sum_b^{\textrm{rank}(M'')} \phi_{m\beta}^b\sum_{m\beta}\psi_{m\beta}\bar{\phi}_{m\beta}^b\nonumber\\
1 &= \sum_b^{\textrm{rank}(M'')} 1  =\textrm{rank}(M''),
\end{align}
where we divided by $\psi_{m\beta}\neq0$ and used the normalistaion relation $\sum_{m\beta}\phi_{m\beta}^b\bar{\phi}_{m\beta}^b=1$. $M''$ being a rank one projector collapses the sums in Eqs.~\eqref{prod_ni} and~\eqref{lambda_ni}, leaving
\begin{align}
|\psi_{m\beta}| = |\phi_{m\beta}^1|.
\end{align}
Substituting this condition into the objective function:

\begin{align}
&||\mathcal{E}_\mathbb{I}-\Gamma||_\diamond/2  =\nonumber\\
&\sum_{a=1}^{\textrm{rank}(A)} \sum_{i,\alpha}\sum_{j\neq i,\beta} |\bar{\phi}^a_{i\alpha}||\phi^a_{j\beta}||\psi_{i\alpha}||\bar{\psi}_{j\beta}|e^{i(\theta_{i\alpha}+\varphi_{i\alpha}-\theta_{j\beta}-\varphi_{j\beta})}\nonumber\\
&\leq\sum_{i,\alpha}\sum_{j\neq i,\beta}|\psi_{i\alpha}|^2|\bar{\psi}_{j\beta}|^2\nonumber\\
%& = \sum_{i}^d\sum_\alpha^d\sum_{\beta}^d\sum_{j\neq i}^d|\psi_{i\alpha}|^2|\bar{\psi}_{j\beta}|^2\\
& = \sum_{i}\sum_\alpha|\psi_{i\alpha}|^2\sum_{\beta}^d\sum_{j\neq i}|\bar{\psi}_{j\beta}|^2\nonumber\\
& = \sum_{i}\sum_\alpha|\psi_{i\alpha}|^2(1-\sum_{\beta}|\bar{\psi}_{i\beta}|^2)\nonumber\\
& = 1-\sum_{i,\alpha,\beta}|\psi_{i\alpha}|^2|\psi_{i\beta}|^2.
\end{align}
We now use the Lagrange multiplier method again, this time optimising over the magnitudes $|\psi_{i\alpha}|$. The Lagrangian is
\begin{align}
\mathcal{L}'=1-\sum_i \sum_\alpha |\psi_{i\alpha}|^2\sum_\beta|\psi_{i\beta}|^2 + \lambda(1-\sum_{i,\alpha}|\psi_{i\alpha}|^2).
\end{align}
Setting the derivative with respect to $|\psi_{k\gamma}|$ to zero yields:
\begin{align}
\lambda|\psi_{k\gamma}| &= \sum_\beta 2|\psi_{k\beta}|^2|\psi_{k\gamma}|\nonumber\\
\lambda &= \sum_\beta 2|\psi_{k\beta}|^2.
\end{align}
We divided by $|\psi_{k\gamma}|\neq0$. Now, summing over $k$ and using the normalisation condition $\sum_{k,\beta}|\psi_{k\beta}|^2=1$ we get 
\begin{align}
\lambda = \frac{2}{d}.
\end{align}
Noticing that the objective function has the form
\begin{align}
||\mathcal{E}_\mathbb{I}-\Gamma||_\diamond/2 &= 1-\sum_i (\lambda/2)^2\nonumber\\
& = 1-d\left(\frac{1}{d^2}\right) = 1-1/d.
\end{align}
\end{proof}
Thus the diamond-norm distance between the two channels is the same as the induced trace-norm distance.

\end{appendix}
%\bibstyle{h-physrev}
%\def\urlprefix{}
%\def\url#1{}
%\def\Eprint#1{}
%\def\Doi{}
%\def\Note{}
%\bibliography{/Users/georgeknee/Documents/paper_library/gck_full_bibliography}
\bibliography{gck_full_bibliography}

\end{document}